\documentclass[twoside,12pt]{article}

\usepackage[margin=1in]{geometry}
\usepackage{amsmath,amssymb,amsthm}
\usepackage{enumitem}
\usepackage{hyperref}
\usepackage{microtype}
\usepackage{booktabs}
\usepackage{tabularx}
\usepackage{placeins}
\usepackage{titling}
\newtheorem{theorem}{Theorem}
\newtheorem{lemma}[theorem]{Lemma}

\newtheorem{corollary}[theorem]{Corollary}
\newtheorem{definition}{Definition}
\newtheorem{assumption}{Assumption}
\newcommand{\SRKESH}{\textsf{SRK-ESH}}

\newcommand{\E}{\mathbb{E}}
\newcommand{\Prb}{\mathbb{P}}

\newcommand{\Pois}{\operatorname{Poisson}}
\newcommand{\Bin}{\operatorname{Binomial}}

\title{Self-Referential $K$-SAT and the Finite Analogue of G\"{o}del’s Incompleteness Theorem}

\author{\normalsize
  Wen Fang$^1$,\quad Xianxian Li$^2$,\quad Jun Liu$^3$,\quad Jie Luo$^1$,\quad Yongxin Tong$^1$,\quad Ke Xu$^{1*}$
}
\date{}

\begin{document}
\maketitle
\vspace{-30pt}
\begin{center}
\footnotesize
$^1$State Key Lab of Complex and Critical Software Environment, Beihang University, Beijing, 100083, China\\
$^2$School of Computer Science and Engineering, Guangxi Normal University, Guilin, 541004, China\\
$^3$School of Mathematics, Taiyuan University of Technology, Taiyuan, 030600, China\\

\end{center}

\vspace{-6pt}
\begin{abstract}
\footnotesize
Self-reference and solution independence are core properties of hard combinatorial instances. We ask whether Boolean \(K\)-SAT can manifest both, establishing a finite combinatorial analogue of G\"{o}del's incompleteness theorems. We formalize solution independence via factorial moments of the satisfying-assignment count. For constant-width random \(K\)-SAT, strong correlations among overlapping assignments disrupt solution independence: small positive Hamming shells make the normalized second moment exponentially large, so the native ensemble cannot reproduce Model-RB-style independence. To overcome this obstruction, we introduce a random CNF ensemble with logarithmic clause width \(K=O(\log N)\) at the subcube-covering threshold \(M=\Theta(N^{2+\varepsilon})\). In this regime the total satisfying-assignment count converges to a Poisson distribution, enabling unsatisfiable and uniquely satisfiable formulas to coexist with positive limiting probabilities. Using one redundant clause and the unique solution, we perform a single-clause replacement whose new clause is falsified by that unique solution, producing a SAT/UNSAT pair with the same unsigned clause--variable incidence graph. We prove a structural irreducibility theorem: each good uniquely satisfiable formula admits, for every local view of at most \(N^c\) clauses (\(0<c<1\) fixed), an unsatisfiable companion sharing that view. Thus no deterministic or bounded-error clause-query evaluator restricted to that scale can distinguish unique satisfiability from unsatisfiability. This result is a finite, propositional analogue of G\"{o}del's self-referential construction; it is formulated within a local subinstance-evaluation model and is not by itself a time-complexity lower bound. We then quantify the induced local--global information gap. Any transcript of \(t=N^{1-\delta}\) clause queries leaves \(N-o(N)\) bits of witness entropy, so the global unique solution remains almost completely hidden from bounded local observation. On typical support graphs, the same self-referential construction preserves the expansion properties of the unsigned incidence graph. Combining this with established size--width, size--degree, and pseudoexpectation trade-offs yields linear Resolution width, linear PC/PCR and SOS degree, and exponential proof size for the unsatisfiable companions; under an explicit partition condition the same conclusion holds for semantic Cutting Planes, and with product-width restrictions for Positivstellensatz. These lower bounds hold uniformly over all support-preserving signings, including replacements chosen after observing the formula and its unique solution. These bounds identify exponential resource requirements in the stated proof systems. Combining structural evidence from solution independence with
an analysis of the limits of available acceleration strategies,
we motivate the self-referential $K$-SAT exhaustive-search
hypothesis (\SRKESH), which rules out a fixed proportional saving
in the exponent for uniform deterministic algorithms deciding
the constructed promise family. We interpret this hypothesis as a quantitative computational projection of the G\"{o}del-style local--global gap. The same self-referential hardness is invariant under quantum computing and delimits a fundamental scaling bottleneck for purely local statistical learning architectures.
\end{abstract}


\footnotetext[1]{E-mail: kexu@buaa.edu.cn}

\setcounter{tocdepth}{2} 
\clearpage
\tableofcontents

\newpage

\section{Introduction}

Self-reference gives rise to important impossibility results across logic, computability, and computational complexity, with important philosophical significance. In classical mathematical logic, self-reference enables a formal system to encode structurally reflexive statements about its own properties. 
Previous work by Xu and Zhou~\cite{xu2025} and subsequent studies on self-referential instances~\cite{li2026, zhou-ds2026, zhou2026} explored a similar direction. Their goal is to construct combinatorial instances where a tiny local symmetric transformation swaps the formula between satisfiable and unsatisfiable states. When localized subinstances fail to differentiate between these two opposing global states, the entire combinatorial object exhibits a strict structural irreducibility, forcing any valid solver to evaluate the global instance in its entirety to determine its satisfiability status.

The inspiration for this line of inquiry traces directly back to G\"{o}del’s incompleteness theorem. G\"{o}del’s construction famously arithmetizes syntax to build a meta-mathematical sentence asserting its own unprovability within a defined formal system~\cite{godel1931}. As emphasized by Budiansky~\cite{budiansky2021}, this theorem historically destroyed Hilbert's ambition to secure all of mathematics through a complete, enclosed formal proof apparatus. Its enduring lesson, however, is creative rather than purely negative: mathematics becomes inexhaustive  because no fixed formal system can contain the entirety of arithmetic truth. While the present paper does not operate in the infinite domain of G\"{o}del's unprovability, it establishes its precise finite combinatorial counterpart. We demonstrate that a Boolean formula can be structurally manipulated using the information of its own unique solution. For each uniquely satisfiable base formula in the stated good event and every local view of at most $N^c$ clauses, with fixed $0<c<1$, there exists an unsatisfiable companion sharing that view. The companion may depend on the view, establishing structural irreducibility of the resulting family under local evaluation.

Recent theoretical developments~\cite{zhou2026} argue that a structural mechanism supporting such constructions is solution independence, a landscape where candidate solutions behave as mutually independent potential witnesses. This property explains why Model RB~\cite{xu2000, xu2025}---which achieves solution independence by increasing its domain size---and classical random graph instances for the clique~\cite{li2026} and dominating set~\cite{zhou-ds2026} problems, naturally support self-referential constructions. By contrast, many canonical NP-complete problems, such as the vertex cover problem, exhibit positive structural correlations among candidate solutions. This structural correlation provides exploitable semantic clues, explaining why localized pruning and structural shortcuts can evade brute-force search in those settings.

The Boolean Satisfiability (SAT) problem occupies a unique position in computer science, renowned not only as the inaugural NP-complete problem~\cite{cook1971}, but more fundamentally as the precise question posed in G\"{o}del’s historic 1956 letter to von Neumann. This historical connection underscores the foundational significance of the problem. This paper focuses on the standard $K$-SAT formulation, which is NP-complete for clause widths $K \ge 3$. Here, we utilize $K$ to denote the maximum clause width, distinguishing between standard fixed-width random $K$-SAT where $K$ is a constant~\cite{franco1985, ding2015}, and logarithmic-width $K$-SAT where $K = O(\log N)$ as a function of the number of Boolean variables $N$~\cite{frieze2005, liu2012}. Classical fixed-width $K$-SAT algorithms exploit short-clause dependencies to improve on exhaustive search~\cite{paturi2005}. Historically, the Strong Exponential-Time Hypothesis (\textsf{SETH}) formalized the conjecture that no fixed proportional exponent saving works uniformly across all fixed clause widths~\cite{calabro2009}. The recent fixed-width counting result of OpenAI~\cite{openai2026eth} prompts a reassessment of this background. Our focus is a different regime: self-referential instances whose width grows with the number of variables. Sections~\ref{sec:independence-acceleration} and~\ref{sec:srk-projection} analyze acceleration in this regime and formulate a hypothesis for the resulting promise family.

The central question is:
\[
  \emph{Can $K$-SAT be both solution-independent and self-referential?}
\]

Our contributions are summarized as follows:

\begin{itemize}
    \item \textbf{Defining Solution Independence:} We define solution independence for Boolean CNF formulas using factorial moments of the satisfying assignment count. This characterizes the standard threshold condition and establishes a stronger Poisson asymptotic convergence. The weakest form used is the critical second-moment condition $\E[X]\to\lambda$ and $\E[X(X-1)]/\E[X]^2\to1$; the strongest form is factorial-moment convergence implying $X\Rightarrow\mathrm{Poisson}(\lambda)$.

    \item \textbf{Constant-Width Clauses:} We prove that standard random $K$-SAT formulas with fixed clause lengths cannot achieve solution independence at the critical scale. High-overlap assignment correlations exponentially inflate the normalized second moment: assignment pairs at small positive Hamming distance contribute an exponential factor that is incompatible with solution independence. This explains why native flat fixed-width random clauses cannot reproduce the Model-RB-style mechanism.

    \item \textbf{Logarithmic-Width CNF Ensemble:} We introduce a random CNF ensemble with logarithmic clause width $K=\lceil(1+\varepsilon)\log_2 N\rceil$ and $M=\Theta(N^{2+\varepsilon})$ clauses at the subcube-covering threshold. We show that the solution count converges asymptotically to a Poisson distribution. At this threshold, unsatisfiable and uniquely satisfiable formulas coexist with positive limiting probability, and typical uniquely satisfiable formulas contain many redundant clauses. This supplies the structural substrate needed for a legal one-clause SAT/UNSAT flip.

    \item \textbf{Self-Referential One-Clause Flips:} Using the unique solution $s$ of a uniquely satisfiable logarithmic-width formula $F$, we replace a redundant clause by a killing clause $B_s$ that is falsified by $s$. The resulting formula $F'$ is unsatisfiable, differs from $F$ in exactly one clause, and has the same unsigned variable--clause incidence graph. This gives an explicit self-referential SAT/UNSAT pair in native Boolean CNF with clause width $O(\log N)$.

    \item \textbf{Structural Irreducibility Theorem:} We prove that the constructed SAT/UNSAT pairs cannot be separated by local subinstance evaluation. For every local view of size at most $N^c$, $0<c<1$, there exists an unsatisfiable companion with exactly the same local view as the uniquely satisfiable base formula. Thus any deterministic or bounded-error adaptive evaluator restricted to at most $N^c$ clause queries cannot decide satisfiability for this promise family. The result is a finite, propositional analogue of G\"{o}del incompleteness; it is formulated within a local subinstance-evaluation model, not as a time-complexity lower bound.

    \item \textbf{Local--Global Information Gap:} We quantify the information-theoretic content of a bounded local transcript. If a deterministic query strategy uses at most $t=N^{1-\delta}$ clause queries, then the transcript leaves $N-o(N)$ bits of witness entropy, and any fixed decoder recovers the global unique solution with probability at most $2^{-(N-o(N))}$. With high probability, the conditional prefix-free Kolmogorov complexity of the unique solution given the transcript remains $N-o(N)-b$. This formalizes the local--global information barrier: sublinear observation exposes only a vanishing fraction of the witness entropy.

    \item \textbf{Proof-Complexity Lower Bounds:} On typical support graphs satisfying a uniform expansion condition, the unsigned clause--variable incidence graph forces linear Resolution width $\Omega(N)$ and hence exponential Resolution refutation size $2^{\Omega(N)}$. The same expansion and pseudoexpectation arguments yield linear degree and exponential monomial size for Polynomial Calculus and PCR, linear SOS degree and exponential SOS monomial size, bounded-product-width Positivstellensatz lower bounds, and, under an explicit partition condition, exponential semantic Cutting Planes lower bounds. These bounds are preserved by the support-preserving self-referential flip and hold uniformly over all unsatisfiable signings on the same support graph, including companions chosen after observing the unique solution.

    \item \textbf{Algorithmic Acceleration and \SRKESH:} We analyze how solution independence limits the transfer of candidate outcomes and how increasing clause width weakens the direct acceleration guarantees supplied by short clauses. We also identify a width-dependent sparsification budget that prevents the recent fixed-width counting analysis from transferring directly to our regime. These observations motivate \SRKESH, a deterministic worst-case hypothesis excluding fixed proportional exponent savings on the constructed promise family, and its interpretation as a quantitative projection of G\"{o}del-style incompleteness onto finite computation. We also propose a dual-evidence criterion for motivating hardness
hypotheses, combining intrinsic structural mechanisms with the
observed limits of available algorithms on the same instance
family and resource scale.

    \item \textbf{Meta-Obstacles for Complexity Theory:} We diagnose the long-standing stagnation in lower-bound theory through two structural constraints: the \emph{Self-Recognition Obstacle} and the \emph{Static-Dynamic Obstacle}. The first provides an interpretive framework connecting the
limits of local syntactic reasoning with global semantic
requirements and with classical barriers in complexity theory. The second explains why purely static mathematical axioms cannot directly yield unconditional temporal lower bounds without a bridge to operational semantics. Together they reframe complexity barriers as logical necessities rather than mere algorithmic failures.

    \item \textbf{Paradigm Comparison Framework:} We introduce a multi-dimensional comparative framework contrasting classical complexity class separation with instance indistinguishability. The former traces to Turing's separation of infinite language classes and measures dynamic execution traces; the latter traces to G\"{o}del's construction of individual self-referential sentences and measures static structural irreducibility under local evaluation. This comparison distinguishes the Turingian and G\"{o}delian
paradigms through their accounts of the source of hardness,
contrasts black-box diagonalization with structural white-box
diagonalization, and clarifies the scope of the relativization
barrier in relation to concrete instance analysis.

    \item \textbf{Hardness Across Different Systems:} We show that self-referential hardness is invariant across changing computing paradigms. In the quantum setting, the same sublinear informational threshold and the need for global semantic analysis preclude structural quantum shortcuts. In machine learning, models operating purely through lossy local statistical compression lack the descriptive representation needed to resolve self-referential instances, delineating a fundamental scaling bottleneck for local learning architectures.
\end{itemize}

\section{Preliminaries}

\subsection{Boolean CNF}

Let $N$ be the number of Boolean variables, and let
\[
  V = \{x_1,\ldots,x_N\}.
\]
A $K$-clause is a disjunction of at most $K$ literals over $V$.
In the random ensemble every clause has exactly $K$ distinct variables and
contains no complementary pair of literals. Clauses are stored in indexed
slots; deleting or replacing a clause refers to one slot, even if repetitions
occur. A $K$-SAT
formula is a conjunction of such clauses. Let $M$ denote the number of clauses.

For a formula $F$, let
\[
  X(F) = \#\{s \in \{0,1\}^N : s \text{ satisfies } F\}
\]
be the number of satisfying assignments. When the formula is random, we write
$X$ for the corresponding random variable.

\subsection{Solution Independence}

For an assignment $\sigma$, define
\[
  I_\sigma(F) = \mathbf{1}\{\sigma \text{ satisfies } F\}.
\]
The weakest form of solution independence used here is the critical
second-moment condition:
\[
  \E[X] \to \lambda, \qquad 0 < \lambda < \infty,
\]
and
\[
  \frac{\E[X(X-1)]}{\E[X]^2} \to 1.
\]
This says that two distinct candidate assignments survive almost as if the
events $\{\sigma \text{ satisfies } F\}$ and $\{\tau \text{ satisfies } F\}$
were independent on average.

A stronger sufficient condition is factorial-moment convergence:
\begin{equation} \label{eq:indep}
\mathbb{E}[(X)_r] \to \lambda^r \quad \text{for every fixed } r \geq 1 \implies X \Rightarrow \text{Poisson}(\lambda) 
\end{equation}
where
\[
  (X)_r = X(X-1)\cdots(X-r+1).
\]
In particular,
\[
  \Prb[X=0] \to e^{-\lambda},
  \qquad
  \Prb[X=1] \to \lambda e^{-\lambda}.
\]
The Poisson form is the cleanest way to express the coexistence of
unsatisfiable and uniquely satisfiable instances with positive limiting
probability. The proof below establishes both the critical second-moment
condition and Poisson convergence directly. It does not require, or infer
from convergence in distribution alone, convergence of all higher factorial
moments.

\subsection{Self-Reference}

For CNF formulas, we use the following operational notion.

\begin{definition}[self-referential SAT/UNSAT flip]
\label{def:flip}
{
Fix a syntactic class $\mathcal{C}$ of CNF formulas, specified by a clause-width
bound and a clause count.
A CNF family in $\mathcal{C}$ admits a \emph{self-referential SAT/UNSAT flip}
if it contains pairs $(F,F')\in\mathcal{C}\times\mathcal{C}$ such that
\[
  \begin{aligned}
  F  &\text{ is uniquely satisfiable},\\
  F' &\text{ is unsatisfiable},
  \end{aligned}
\]
and $F'$ is obtained from $F$ by a \emph{legal one-step transformation},
defined as the replacement of a single clause of $F$ by a single clause of
$\mathcal{C}$. Equivalently, $F$ and $F'$ differ in exactly one clause and
both lie in $\mathcal{C}$.}
\end{definition}

In the logarithmic-width construction, $\mathcal{C}$ is the class of $K$-SAT
formulas with $M$ clauses where $K=\lceil(1+\varepsilon)\log_2 N\rceil$ and
$M=\Theta(N^{2+\varepsilon})$, and the legal transformation is a single
$K$-clause replacement.

\subsection{Random \texorpdfstring{$K$}{K}-SAT Ensemble}

The standard random $K$-SAT ensemble used below is the following. Each clause
is sampled independently by choosing $K$ distinct variables uniformly and then
choosing one of the $2^K$ sign patterns uniformly. Repeated clauses are
allowed; all probability calculations below use this sampling-with-replacement
convention. Asymptotic statements concern sufficiently large $N$ that $K\le N$.

For a fixed assignment, one random $K$-clause is falsified with probability
\[
  p = 2^{-K},
\]
and satisfied with probability
\[
  q = 1-p.
\]

\section{Native Constant-Width Clauses}

This section proves that the natural fixed-width random-clause mechanism does
not yield solution independence.

\begin{theorem}[native fixed-width obstruction]
Fix a constant $K \ge 2$. Let $F_N$ be a random $K$-SAT formula on $N$ Boolean
variables with $M=M(N)$ clauses chosen so that
\[
  \E[X] = 2^N(1-2^{-K})^M = \Theta(1).
\]
Then the native random-clause ensemble is not solution-independent. More
precisely,
\[
  \frac{\E[X(X-1)]}{\E[X]^2}
\]
does not converge to $1$; in fact the contribution from assignment pairs at
small positive Hamming distance is exponentially large.
\end{theorem}

\begin{proof}
Let
\[
  p = 2^{-K}, \qquad q = 1-p.
\]
Let $\sigma$ and $\tau$ be assignments at Hamming distance $h=\beta N$. A
random clause is falsified by both assignments only if every selected variable
lies in a coordinate where the two assignments agree, and the signs match
their common values. Thus
\[
  \Prb[\sigma,\tau \text{ both satisfy one clause}]
    = 1 - 2p + p(1-\beta)^K + o(1).
\]
{
The error term is uniform over $\beta\in[0,1]$: writing the agreement set as
$A\subseteq[N]$ with $|A|=(1-\beta)N$, the exact probability is
$\binom{|A|}{K}/\binom{N}{K}\cdot 2^{-K}$, which differs from $(1-\beta)^K p$
by $O(K^2/N)$ uniformly. The marginals are exactly $q$. Since $K$ is fixed and $M=\Theta(N)$,
the $O(1/N)$ error changes the logarithm of the pair probability over all
clauses by only $O(1)$, and hence does not affect the exponential rate.}

The one-clause correlation ratio is
\begin{equation} \label{eq:corre}
R_K(\beta) = \frac{1 - 2p + p(1 - \beta)^K}{q^2} 
\end{equation}
At the first-moment critical scale,
\[
  M = c_K N + O(1),
  \qquad
  c_K = \frac{\ln 2}{-\ln q}.
\]
The normalized contribution of ordered pairs at distance $\beta N$ has
exponential rate
\[
  \Phi_K(\beta)
    =
    H(\beta) - \ln 2 + c_K \ln R_K(\beta),
\]
where $H(\beta)$ is the binary entropy function with natural logarithms.

At $\beta=0$, the diagonal balance gives $\Phi_K(0)=0$. For small positive
$\beta$,
\[
  H(\beta) = \beta \ln(1/\beta) + \beta + O(\beta^2).
\]
From Eq.~\eqref{eq:corre},
\[
  \ln R_K(\beta)
    =
    -\ln q - \frac{pK}{q}\beta + O(\beta^2).
\]
Thus,
\[
  \Phi_K(\beta)
    =
    \beta \ln(1/\beta) - C_K\beta + O(\beta^2)
\]
where $C_K=c_KpK/q-1$. For all sufficiently small fixed $\beta>0$, the term
$\beta\ln(1/\beta)$ dominates, and hence $\Phi_K(\beta)>0$.

Thus a Hamming shell of small positive radius contributes
\[
  \exp\{\Phi_K(\beta)N + o(N)\}
\]
to the normalized second moment. This is exponentially larger than the
independent value. Hence the second-moment independence condition fails.
\end{proof}

The proof isolates the obstruction. Nearby assignments disagree on only a
small fraction of variables, while a constant-width clause observes only
$K=O(1)$ positions. Most clauses therefore miss the disagreement and treat the
two assignments almost identically.

\begin{corollary}
Native flat fixed-width random clauses cannot reproduce the Model-RB-style
solution-independence mechanism.
\end{corollary}

\section{Logarithmic-Width CNF}

\subsection{Poisson Solution Counting and Pair-Collision Bounds}

We now allow the clause width to grow logarithmically:
\[
  K = O(\log N).
\]
This is still ordinary Boolean CNF, but the clause is wide enough to detect
high-overlap assignments at the scale required by the second moment.

Fix constants
\[
  \varepsilon > 0, \qquad \lambda > 0.
\]
Let
\[
  K = \left\lceil (1+\varepsilon)\log_2 N \right\rceil,
  \qquad
  p = 2^{-K},
  \qquad
  q = 1-p.
\]
Generate $M$ random $K$-clauses independently as in Section~2.4, where
\[
  M
    =
    \left\lfloor
      \frac{N\ln 2 - \ln \lambda}{-\ln(1-p)}
    \right\rfloor.
\]
Since
\[
  p = N^{-(1+\varepsilon)+o(1)},
\]
we have
\begin{equation}\label{eq:claus}
  M
    =
    (1+o(1))2^K(N\ln 2-\ln \lambda)
    =
    \Theta(N^{2+\varepsilon}).
\end{equation}
Every clause has width $O(\log N)$.

\paragraph{Pair-collision estimate.}
For two assignments at Hamming distance $h$, put
\[
 \rho_N(h)=\frac{\binom{N-h}{K}}{\binom NK},
\]
with the numerator interpreted as zero when $N-h<K$. A random clause is
falsified by both assignments with probability $p\rho_N(h)$.

\begin{lemma}[summability of pair collisions]\label{lem:pair}
If $C_N=N\ln 2+O(1)$, then
\begin{equation}\label{eq:pair-sum}
 2^{-N}\sum_{h=1}^N\binom Nh\exp\{C_N\rho_N(h)\}=1+o(1).
\end{equation}
\end{lemma}
\begin{proof}
For every $h$, $\rho_N(h)\le(1-h/N)^K\le e^{-Kh/N}$.
Choose $\eta\in(0,1)$ sufficiently small that
$(1+\varepsilon)(1-\eta/2)>1$, and choose $\gamma>0$ smaller than the
excess over $1$. Write $H(u)=-u\ln u-(1-u)\ln(1-u)$.
Choose $b\in(0,1/2)$ so small that
$H(b)<(\ln2)(1-e^{-\eta})$.
Finally set $\beta=-(1+\varepsilon)\log_2(1-b)>0$, and choose
$0<a<\min\{1/2,\beta/2\}$. Put $d_N=N^{-a}$.
We split the sum into four ranges.

If $1\le h\le\eta N/K$, then
\[
 1-\rho_N(h)\ge 1-e^{-Kh/N}
 \ge (1-\eta/2)Kh/N.
\]
Consequently, uniformly in this range and for all sufficiently large $N$,
\[
 2^{-N}e^{C_N\rho_N(h)}
 \le C\exp\{-(1+\gamma)h\ln N\}.
\]
Its total contribution is at most
$C\sum_{h\ge1}N^{-\gamma h}/h!=o(1)$.

If $\eta N/K<h\le bN$, then $\rho_N(h)\le e^{-\eta}$.
The entropy bound for binomial coefficients gives a contribution at most
\[
 (N+1)\exp\{NH(b)-N\ln2+C_Ne^{-\eta}\}=e^{-\Omega(N)}.
\]

If $bN<h<(1/2-d_N)N$, then
$\rho_N(h)\le(1-b)^K\le N^{-\beta}$.
For $Z\sim\Bin(N,1/2)$, the binomial lower-tail bound gives a contribution
at most
\[
 e^{O(N^{1-\beta})}\Prb[Z\le N/2-Nd_N]
 \le\exp\{O(N^{1-\beta})-2N^{1-2a}\}=o(1).
\]
Here $2a<\beta$ and $a<1/2$ are both used.

For the remaining range $h\ge(1/2-d_N)N$,
\[
 \rho_N(h)\le(1/2+d_N)^K
 =p(1+2d_N)^K=p(1+o(1)).
\]
Since $Np=o(1)$, the exponential factor is uniformly $1+o(1)$.
This range has binomial probability $1-o(1)$ and therefore contributes
$1+o(1)$. The four ranges cover every $h\ge1$ and prove the claim.
\end{proof}

\begin{theorem}[Poisson solution count]\label{thm:poisson}
For the logarithmic-width ensemble above,
\[
 d_{\mathrm{TV}}(\mathcal L(X),\Pois(\lambda))\longrightarrow0,
 \qquad
 \frac{\E[X(X-1)]}{\E[X]^2}\longrightarrow1.
\]
In particular,
\[
 \Prb[X=0]\to e^{-\lambda},\qquad
 \Prb[X=1]\to\lambda e^{-\lambda}.
\]
\end{theorem}
\begin{proof}
The first moment is exact: $\E[X]=2^Nq^M\to\lambda$.
The choice of $M$ and $Np=o(1)$ give
\[
 Mp=N\ln2-\ln\lambda+o(1),\qquad Mp^2=o(1).
\]
For fixed $M$, the exact second factorial moment is
\[
 \frac{\E[X(X-1)]}{(2^Nq^M)^2}
 =2^{-N}\sum_{h=1}^N\binom Nh
 \left(\frac{1-2p+p\rho_N(h)}{q^2}\right)^M.
\]
Uniformly in $h$, the logarithm of the last factor is
$Mp\rho_N(h)+O(Mp^2)$. Lemma~\ref{lem:pair} proves the stated
second-moment convergence.

To obtain the distributional limit, first replace the fixed number of
clauses by $L\sim\Pois(m)$, where $m=M+o(1/p)$, and let $Y_m$ count the
surviving assignments. There are $Q_N=2^K\binom NK$ possible clauses;
their multiplicities are independent Poisson variables of mean $m/Q_N$.
Writing $I_\sigma$ for survival of assignment $\sigma$, we have
\[
 \pi_\sigma:=\E I_\sigma=e^{-mp},\qquad
 \mu_m:=\E Y_m=2^Ne^{-mp}\to\lambda.
\]
Conditioning on $I_\sigma=1$ is achieved by setting to zero the
multiplicities of every clause falsified by $\sigma$ and leaving all other
multiplicities unchanged. This coupling can only increase every other
survival indicator. Thus the indicators are positively related.

The Stein--Chen bound~\cite{arratia1989} for positively related indicators
\cite[Theorem~5.1]{chen-rollin2013} states
\[
 d_{\mathrm{TV}}(\mathcal L(Y_m),\Pois(\mu_m))
 \le \min\{1,\mu_m^{-1}\}
 \left(\operatorname{Var}(Y_m)-\mu_m
       +2\sum_\sigma\pi_\sigma^2\right).
\]
For a pair at distance $h$, its joint survival probability is
$e^{-mp(2-\rho_N(h))}$. Hence Lemma~\ref{lem:pair} gives
\[
 \frac{\E[Y_m(Y_m-1)]}{\mu_m^2}
 =2^{-N}\sum_{h=1}^N\binom Nh e^{mp\rho_N(h)}=1+o(1).
\]
Also $\sum_\sigma\pi_\sigma^2=2^{-N}\mu_m^2=o(1)$, so the bound tends
to zero. This proves the Poisson limit in the Poissonized model.

To return to exactly $M$ clauses, choose $a_N\to\infty$ with
$a_N\sqrt{Mp^2}\to0$ and $a_N/\sqrt M\to0$; for example,
$a_N=(Mp^2)^{-1/4}$ has both properties here. Set
$m_\pm=M\pm a_N\sqrt M$. Use one infinite sequence of independent random
clauses and Poisson counts $L_\pm\sim\Pois(m_\pm)$ independent of that
sequence. Chebyshev's inequality gives
\[
 \Prb[L_-\le M\le L_+]\to1.
\]
On this event monotonicity yields $Y_{m_+}\le X\le Y_{m_-}$.
Since $p|m_\pm-M|=o(1)$, both bounding counts converge to
$\Pois(\lambda)$. Their distribution functions sandwich that of $X$ at
every half-integer. Thus $X\Rightarrow\Pois(\lambda)$.
For integer-valued laws with a probability distribution as limit, pointwise
convergence of the masses implies convergence in total variation, by
truncation to a finite set. This proves the theorem.
\end{proof}

The proof can also be viewed as a random subcube-covering argument. Each
$K$-clause forbids exactly one codimension-$K$ subcube of the Boolean cube.
At the density above, the expected number of uncovered vertices is constant,
and their total number has a Poisson limit. No spatial point-process
convergence is asserted here.

\subsection{Self-Referential SAT/UNSAT Flips}

The Poisson limit gives both unsatisfiable and uniquely satisfiable formulas
with positive limiting probability. To obtain an explicit self-referential
pair, we also need a redundant clause in a unique instance.

Let $W$ be the number of assignments that falsify exactly one clause of $F$.
For a fixed assignment, the number of falsified clauses is distributed as
\[
  \Bin(M,p).
\]
Thus,
\[
  \E[W]
    =
    2^NMp(1-p)^{M-1}
    =
    (1+o(1))\lambda Mp
    =
    \Theta(N).
\]
From Eq.~\eqref{eq:claus},
\[
  M = \Theta(N^{2+\varepsilon}).
\]
Markov's inequality gives
\[
  \Prb[W\ge M] \le \frac{\E[W]}{M} = o(1).
\]

\begin{lemma}[redundant clauses]\label{lem:redundant}
With probability tending to $\lambda e^{-\lambda}$, a sampled formula
from the logarithmic-width ensemble is uniquely satisfiable and has at least
$M-N\log N$ redundant clause slots. In particular, conditional on unique
satisfiability, a $1-o(1)$ fraction of its clauses are redundant with
probability $1-o(1)$.
\end{lemma}

\begin{proof}
By Theorem~\ref{thm:poisson},
\[
  \Prb[X=1] \to \lambda e^{-\lambda} > 0.
\]
Markov's inequality gives $\Prb[W>N\log N]=O(1/\log N)=o(1)$.
Thus $\Prb[X=1,\ W\le N\log N]\to\lambda e^{-\lambda}$.

Suppose $F$ is uniquely satisfiable. A clause $D$ is essential if $F-D$ has a
satisfying assignment that is not a satisfying assignment of $F$. For every
essential clause there exists an assignment satisfying all clauses except that
one. Such an assignment is counted by $W$. Therefore, the number of essential
clauses is at most $W$. On the event $W\le N\log N$, at least $M-N\log N$ clauses are
redundant. Since $N\log N=o(M)$, the conclusion follows.
\end{proof}

\begin{theorem}[one-clause logarithmic-width flip]\label{thm:flip}
With probability tending to $\lambda e^{-\lambda}$, a sampled
formula $F$ admits a legal one-clause replacement producing a formula $F'$ such
that
\[
  F \text{ is uniquely satisfiable}, \qquad
  F' \text{ is unsatisfiable},
\]
and both $F$ and $F'$ have clause width $O(\log N)$.
\end{theorem}

\begin{proof}
Let $F$ be a uniquely satisfiable formula with unique satisfying assignment
$s$, and let $D$ be a redundant clause. Then $F-D$ is still uniquely satisfied
by $s$.

Choose $J$ to be the variable set of $D$. Define a clause $B_s$ by
\[
  B_s = \bigvee_{i\in J}\ell_i,
\]
where
\[
  \ell_i =
  \begin{cases}
    x_i, & s_i=0,\\
    \neg x_i, & s_i=1.
  \end{cases}
\]
The assignment $s$ falsifies every literal in $B_s$, so it falsifies $B_s$.
Set
\[
  F' = (F-D)\wedge B_s.
\]
Since $F-D$ has only the satisfying assignment $s$, and $B_s$ kills exactly
this remaining assignment among the solutions of $F-D$, the formula $F'$ is
unsatisfiable. The two formulas differ by one legal $K$-clause, and
$K=O(\log N)$. They have exactly the same unsigned variable--clause
incidence graph, since $B_s$ uses the variables of $D$.
\end{proof}

This theorem supplies the desired self-referential pair in native Boolean CNF
with logarithmic clause width.

\subsection{Structural Irreducibility Theorem}

We now adapt the notion of reducibility used for self-referential dominating
set instances~\cite{zhou-ds2026} to the present SAT setting. The point is to formalize
the following local subinstance-evaluation question: can a sublinear part of
the formula contain enough information to determine the global satisfiability
status?

\begin{definition}[local subinstance reducibility for SAT]
\label{def:local}
A formula is accessed as an array of $M$ clauses. A clause query specifies
an index and returns the variables and signs in that clause. A deterministic
adaptive evaluator may choose its next index and its stopping decision only
from its fixed program, $N,K,M$, and the clauses already returned. Its
transcript contains the queried indices and their answers; no other
input-dependent information is supplied for free. A randomized evaluator
also has input-independent random coins.

A local view of size at most $N^c$, $0<c<1$, is the restriction to at most
$N^c$ clause slots, together with their variable labels and literal signs.
A family $\mathcal F_N$ is reducible at this scale if an evaluator using at
most $N^c$ clause queries correctly decides every member. It is irreducible
at this scale if no such evaluator exists. A sufficient condition is that
on some member, every possible such transcript also occurs on a member of
opposite satisfiability status. Comparisons of identical restrictions always
refer to the same clause indices.
\end{definition}

If $H$ contains at most $N^c$ evaluated clauses, then the variables appearing
in those clauses contribute at most
\[
  K N^c = O(N^c\log N)
\]
seen variables. Thus, every $N^c$-size local subinstance sees at most
$O(N^c\log N)=O(N^{1-\delta})$ 
variables for any fixed $0<\delta<1-c$. The next lemma strengthens this observation:
uniformly over all such small observed variable sets, only a vanishing
fraction of the clauses even touch the observed region.

\begin{lemma}[uniform small-boundary estimate]
\label{lem:usbe}
Fix $0<c<1$ and a constant $a>0$. With probability $1-o(1)$ over the
logarithmic-width ensemble, every set $U\subseteq[N]$ with
\[
  |U|\le aN^c\log N
\]
is touched by at most $M/8$ clauses.
\end{lemma}

\begin{proof}
The assertion is immediate for $U=\varnothing$. For a fixed nonempty set $U$ of size $u$, let $T_U$ be the
number of clauses whose
variable set intersects $U$. Since clauses are sampled independently,
$T_U$ is binomial with parameters $M$ and
\[
  \alpha_U
  =
  1-\frac{\binom{N-u}{K}}{\binom NK}.
\]
For $u\le aN^c\log N$ and $K=O(\log N)$,
\[
  \alpha_U
  \le
  \frac{Ku}{N-K+1}
  \le
  C N^{c-1}(\log N)^2
  =
  o(1),
\]
where $C$ depends only on $a$ and $\varepsilon$. Hence
\[
  \mu_U:=\E[T_U]\le \alpha_U M=o(M).
\]
For all sufficiently large $N$, $\mu_U\le M/16$. The standard binomial
Chernoff bound gives
{
\[
  \Prb[T_U\ge M/8]
  \le
  \left(\frac{e\mu_U}{M/8}\right)^{M/8}
  =
  \exp\!\left\{\frac{M}{8}\bigl(\log(8e)-\log(M/\mu_U)\bigr)\right\}.
\]

Since $\mu_U=M\alpha_U$, we have $M/\mu_U=1/\alpha_U$, and from
$\alpha_U\le CN^{c-1}(\log N)^2$ it follows that
\[
  \log(M/\mu_U)
  =
  \log(1/\alpha_U)
  \ge
  (1-c)\log N-2\log\log N-\log C
  =
  \Omega(\log N).
\]
Substituting back,
\[
  \Prb[T_U\ge M/8]\le \exp\{-\Omega(M\log N)\}.
\]}
The number of possible sets $U$ with $|U|\le aN^c\log N$ is at most
\[
  \sum_{u\le aN^c\log N}\binom Nu
  \le
  \exp\{O(N^c(\log N)^2)\}.
\]
Since
\[
  M\log N=\Theta(N^{2+\varepsilon}\log N)
\]
dominates $N^c(\log N)^2$ for every $c<1$, the union bound proves the claim.
\end{proof}

\begin{theorem}[irreducibility of the native logarithmic-width pairs]
\label{thm:inlp}
Fix $0<c<1$. With positive limiting probability, a formula $F$ sampled from
the logarithmic-width ensemble admits the following simultaneous property. For
every local subinstance $H$ of size at most $N^c$, there exists a formula
$F'$ such that
\[
  H(F)=H(F'),
\]
but
\[
  F\in \mathrm{UniqueSAT},
  \qquad
  F'\in \mathrm{UNSAT}.
\]
{
The quantifier order is $\forall F\;(\text{in the good event})\;\forall H\;\exists F'$.
The partner $F'$ may depend on $H$. The statement applies to adaptive
evaluators as in Definition~\ref{def:local}: induction on the queries gives
the same transcript on both inputs. It does not allow choosing queries by
examining unqueried parts of the input for free.}
Thus, any promise family containing these base formulas and all the constructed
companions is irreducible in the specified local subinstance-evaluation sense.
\end{theorem}

\begin{proof}
We prove the claim by a finite diagonalization against deterministic local evaluators: preserve the evaluator's observations while reversing the global satisfiability status. We first construct a uniquely satisfiable base formula for which every permitted local view admits such an unsatisfiable companion.

Let $F$ be sampled from the logarithmic-width ensemble, and let $W$ be the
number of assignments that falsify exactly one clause. As in Section~4.2,
\[
  \E[W]=\Theta(N),
  \qquad
  M=\Theta(N^{2+\varepsilon}).
\]
Thus
\[
  \Prb[W\ge M/4]\le \frac{4\E[W]}{M}=o(1).
\]
By Theorem~\ref{thm:poisson},
\[
  \Prb[X=1]\to \lambda e^{-\lambda}>0.
\]
Combining this with the uniform small-boundary estimate, with positive
limiting probability, all three events hold:
\[
  X=1,\qquad W<M/4,\qquad
  \text{every }|U|\le aN^c\log N\text{ touches at most }M/8\text{ clauses},
\]
where $a$ is chosen large enough to cover every $N^c$-size local subinstance.
Condition on such a formula $F$, and let $s$ be its unique satisfying
assignment.

As in the proof of Lemma~\ref{lem:redundant}, every essential clause is witnessed by an
assignment that falsifies exactly one clause. Hence, the number of essential
clauses is at most $W$. Since $W<M/4$, at least $3M/4$ clauses are redundant.

Let $H$ be any local subinstance of size at most $N^c$, and let $U_H$ be the
set of variables seen by $H$. By definition,
\[
  |U_H|\le aN^c\log N.
\]
By the small-boundary estimate, at most $M/8$ clauses touch $U_H$. Thus, 
there exists a redundant clause $D$ whose variable set is disjoint from
$U_H$. Indeed, the number of redundant clauses disjoint from $U_H$ is at
least
\[
  M-W-M/8 > 5M/8.
\]
In particular, $D$ is not evaluated by $H$.

Since $|U_H|=o(N)$ and $K=O(\log N)$, for all sufficiently large $N$ there is
a set
\[
  J\subseteq [N]\setminus U_H,
  \qquad |J|=K.
\]
Define the killing clause
\[
  B_s=\bigvee_{i\in J}\ell_i,
  \qquad
  \ell_i=
  \begin{cases}
    x_i, & s_i=0,\\
    \neg x_i, & s_i=1.
  \end{cases}
\]
The clause $B_s$ uses only variables outside the evaluated region, and $s$
falsifies $B_s$.

Now, set
\[
  F'=(F-D)\wedge B_s.
\]
Because $D$ is redundant, $F-D$ still has the unique satisfying assignment
$s$. Because $B_s$ falsifies $s$, the formula $F'$ is unsatisfiable. The
operation removes a redundant clause outside the evaluated region and adds a
new clause outside the evaluated region, so that all clausal, variable, literal, and
incidence information visible to $H$ remains unchanged:
\[
  H(F)=H(F').
\]
Thus, $H$ cannot determine whether the full formula is satisfiable. To make the diagonalization explicit, fix a deterministic evaluator $A$ with at most $N^c$ clause queries and let $H$ be its transcript on $F$. If $A$ outputs UNSAT, it errs on $F$; otherwise, the construction gives an unsatisfiable $F'$ with the same transcript. Induction on the adaptive queries then gives the same SAT output on $F'$, so $A$ errs there.
To minimize changes, $B_s$ can use the same variables as $D$, with the literal polarity determined by the rule above.
Since the same conditioned formula works for every $H$ of size at most $N^c$, and since
$0<c<1$ is arbitrary, no local subinstance of size $N^c$ contains sufficient
information to determine satisfiability for the self-referential family.
In other words, local information is always insufficient, implying that solving formulas in this family inherently requires global information beyond each fixed $N^c$-clause observation budget in the stated
query model.
\end{proof}

This theorem is the SAT analogue of the irreducibility argument in the
dominating-set paper~\cite{zhou-ds2026}: a sublinear local subinstance can be preserved while a
local transformation outside it flips the global solution status. The result
should be interpreted within the framework of local subinstance evaluation, not as a
time-complexity lower bound on SAT algorithms.

Unlike standard time complexity, irreducibility reflects an essential structural feature of problems, rather than algorithmic runtime behavior. Significantly, irreducibility is an intrinsic structural property that holds regardless of whether a solver relies on DPLL, CDCL, resolution, or local search strategies. For any fixed constant $c<1$, any approach confined only to local subinstances of size $N^c$ cannot settle the satisfiability of this problem family. This is because the same local configuration is consistent with both uniquely satisfiable and unsatisfiable formulas. Here the restriction concerns the total observed input, not the number of literals in one intermediate deduction or the information used in one step. An algorithm can accumulate global information through many local steps.

Structural irreducibility is, in this sense, also information irreducibility. For every local observation $H$ within the stated $N^c$-clause budget, the construction provides a uniquely satisfiable formula $F$ and an unsatisfiable companion $F'$ with $H(F)=H(F')$. Thus, the information revealed by $H$ is compatible with opposite global satisfiability states and cannot determine which state holds. The structural obstruction is therefore an informational one: the observed local structure does not contain sufficient information to resolve global satisfiability. Structural irreducibility and information irreducibility describe the same local--global obstruction from two perspectives.

\begin{theorem}[clause-query lower bounds]\label{thm:queries}
Fix a uniquely satisfiable formula $F=(D_1,\ldots,D_M)$ with unique
assignment $s$. Let $\mathcal R$ be its redundant clause indices, and put
$R=|\mathcal R|$. For $j\in\mathcal R$, let $F^{(j)}$ replace $D_j$ by
the clause on the same variables that $s$ falsifies. Define the promise family
\[
 \mathcal P(F)=\{F\}\cup\{F^{(j)}:j\in\mathcal R\}.
\]
Every deterministic clause-query algorithm correct on $\mathcal P(F)$
must query all $R$ indices in $\mathcal R$ on input $F$.
Every randomized algorithm using at most $t$ queries on each run and having
error at most $1/3$ on every promised input satisfies $t\ge R/3$.
For the formulas in Lemma~\ref{lem:redundant}, these bounds are respectively
$M-N\log N$ and $\Omega(M)=\Omega(N^{2+\varepsilon})$.
\end{theorem}
\begin{proof}
Every $F^{(j)}$ is unsatisfiable by Theorem~\ref{thm:flip}.
If a deterministic algorithm stops on $F$ without querying $j\in\mathcal R$,
all its answers, subsequent queries, and its output are the same on
$F^{(j)}$. It therefore errs on one input.

For the randomized claim, compare input $F$ with input $F^{(J)}$, where
$J$ is uniform on $\mathcal R$ and independent of the algorithm's coins.
Couple the runs using identical coins. Until index $J$ is queried their
transcripts coincide. For each fixed coin sequence, the run on $F$ queries
at most $t$ indices, so the probability of any discrepancy is at most $t/R$.
Thus the transcript laws $P_0,P_1$ satisfy
\begin{equation}\label{eq:tv_bound}
 d_{\mathrm{TV}}(P_0,P_1)\le\min\{1,t/R\}.
\end{equation}
The optimal success probability for distinguishing them with equal priors
is $(1+d_{\mathrm{TV}}(P_0,P_1))/2$. Worst-case error at most $1/3$ implies
average success at least $2/3$, and hence $t\ge R/3$.
\end{proof}

The family obtained by taking the union of $\mathcal P(F)$ over the good
base formulas has these worst-case lower bounds. Reading every clause uses
$M$ queries, after which computation is unrestricted, so its deterministic
and bounded-error classical query complexities are both $\Theta(M)$.
This strengthens the sublinear-view obstruction, but measures access to the
input rather than the cost of computation after the input has been read.
With a standard explicit binary encoding, input length is
$\Theta(MK\log N)$; the parameter $N$ is the number of variables, not the
input length.

\section{Finite Analogues of G\"{o}del Incompleteness}

\subsection{Analogue of G\"{o}del’s First Incompleteness Theorem}

The construction presented in the preceding sections bears an interesting structural and conceptual connection to G\"{o}del’s first incompleteness theorem~\cite{godel1931}. G\"{o}del’s proof famously converts syntactic facts about formulas and proofs into arithmetic facts about natural numbers, subsequently constructing a meta-mathematical sentence that references its own unprovability within the system~\cite{godel1931}. In its informal expression, the G\"{o}del sentence declares: ``This sentence is unprovable.'' For a sufficiently expressive, consistent, and effectively axiomatized formal theory $T$, this sentence remains inherently undecidable inside $T$. For a comprehensive introductory exposition of G\"{o}del’s incompleteness theorems and the extensions of classical mathematical logic to sequences of 
formal theories, see~\cite{li2010}.

Our construction shares this exact self-referential configuration but targets a distinct semantic endpoint. Let $F$ be a uniquely satisfiable logarithmic-width formula, and let $s$ be its unique satisfying assignment. The killing clause $B_s$ is defined from $s$ such that $s$ explicitly falsifies $B_s$:
\[
B_s = \bigvee_{i \in J} \ell_i
\]
Replacing a redundant clause $D$ by $B_s$ yields the companion instance:
\[
F' = (F - D) \wedge B_s
\]
Consequently, $F'$ is obtained by exploiting internal semantic information about $F$ itself. Informally, this companion instance states: ``This formula’s unique solution is forbidden.'' 

This operation establishes a precise finite self-reference.
While prior investigations realized self-reference through symmetric mappings over infinite sets~\cite{xu2025}, the method employed here directly mirrors G\"{o}del’s intrinsic self-reference. 
G\"{o}del’s theorem isolates mathematical truth from formal provability within fixed axiomatic systems, while our main result separates global satisfiability from local evaluation. Both paradigms reveal that local information and internal syntactic reasoning are structurally insufficient to determine the global attributes of self-referential objects, driven by the inescapable divide between syntax and semantics. The explicit structural alignment is formalized in Table~\ref{tab:godel_comparison}.

\begin{table}[h]
\centering
\caption{Mapping Between G\"{o}del Incompleteness and Present Finite Analogue}
\label{tab:godel_comparison}
\begin{tabular}{l|l}
\hline
\textbf{G\"{o}del Incompleteness} & \textbf{Present Finite Analogue} \\ \hline
Infinite Domains & Finite Combinatorial Domains \\
First-Order  Logic & Propositional Logic \\
Arithmetical Statement & Boolean $K$-SAT Formula \\
Self-Reference via Unprovability & Self-Reference via Solution Exclusion \\
Finite Proof & Local Evaluation \\
Truth Undecidable via Fixed Systems & SAT Undecidable via Local Views \\
Systemic Incompleteness & Structural Irreducibility \\ \hline
\end{tabular}
\end{table}

In this exact sense, our framework establishes a finite analogue of G\"{o}del’s first incompleteness theorem: \textit{Local views are inherently incomplete for global SAT/UNSAT distinguishability}. 
Budiansky's book~\cite{budiansky2021} stresses two points that are useful for interpreting this analogy. First, Gödel's result maps a limit for formal mathematics, contradicting Hilbert’s vision of a complete formal system for all mathematical truths. Second, Gödel himself did not view incompleteness merely as collapse; he took it as evidence that mathematics is inexhaustive. As recounted in~\cite{budiansky2021}, he remarked that: ``Hilbert was mistaken on one point, namely in limiting the definition of mathematical truth to consistency of formal systems derived from axioms. The axioms too are part of mathematical truth, but of a kind that defies formalism altogether, accessible only via human intuition.” From this perspective, formal reasoning within a fixed system is essentially a local view of mathematics, and hence is inherently incomplete.

Building upon these logical foundations, Tarski’s undefinability theorem~\cite{tarsk1956} 
first formalized the gap between syntax and semantics, proving that a consistent 
formal system cannot define its own truth predicate without triggering fatal paradoxes. 
Subsequently, Chaitin’s incompleteness theorem~\cite{chaitin1974} generalized this structural 
barrier via algorithmic information theory, proving that any fixed axiomatic system is bounded 
by Kolmogorov complexity constraints and cannot certify string randomness past a specific 
descriptive threshold. Concurrently, Paris and Harrington~\cite{paris1977} realized 
incompleteness within core mathematics via the finite Ramsey theorem—a purely combinatorial statement unprovable in Peano Arithmetic. Friedman~\cite{friedman1998} further broadened this independence, demonstrating finite claims that surpass the deductive bounds of ZFC. Yet, while this classical lineage of independence results across proof theory remains completely isolated from operational algorithmics, our self-referential $K$-SAT framework offers a unique computational counterpart, grounding finite structural incompleteness directly within practical solver dynamics and complexity metrics.

For a more analytical focus, this boundary of internal recognition mirrors the 
insight of Laozi regarding the ultimate reality~\cite{lau1963}: ``\emph{The way 
that can be spoken of is not the constant way.}'' It equally echoes Wittgenstein’s 
celebrated closing maxim~\cite{wit1961}: ``\emph{Whereof one cannot speak, 
thereof one must be silent.}'' Together, G\"{o}del's incompleteness theorem and its 
subsequent extensions, these enduring philosophical traditions, and our present 
results jointly address the meta-question of whether a resource-bounded system 
can achieve complete global self-recognition from the inside. 
In this light, Turing’s classical proof of the uncomputability of 
the Halting Problem emerges as a specific manifestation of this 
structural constraint: a computational entity can never fully certify 
its own long-term evolution via local evaluation. In other words, any 
system tasked with complete self-recognition is bound to face 
a structural blindness---a limitation we designate as the 
\textit{Self-Recognition Obstacle}.

\subsection{Analogue of G\"{o}del’s Second Incompleteness Theorem}

G\"{o}del’s second incompleteness theorem establishes that any consistent formal system capable of expressing elementary number theory cannot prove its own consistency using its own internal deductive rules. An operational parallel manifests directly within computational complexity under our framework: as demonstrated by our structural irreducibility theorem, any algorithmic reasoning method restricted to local subinstances of size $N^c$ (where $0 < c < 1$) is incapable of determining the global satisfiability of self-referential formulae. At the critical threshold, this limitation arises because a structurally identical local observation remains perfectly compatible with two mutually exclusive global semantic states: the instance being uniquely satisfiable or entirely unsatisfiable.

Because identical local evidence cannot resolve these conflicting global scenarios, any algorithm relying exclusively on sublinear local information lacks the necessary context to validate its own correctness. This fundamental deficit creates what we define as the \textit{Self-Recognition Obstacle} in Section~5.1: the algorithm cannot verify whether its local logical deduction remain globally consistent with the true semantic state of the entire system. This structural barrier yields a finite analogue of G\"{o}del’s second incompleteness theorem.

\begin{theorem}[Finite Analogue of G\"{o}del's Second Incompleteness]
Any correct algorithm or formal reasoning system capable of solving self-referential $K$-SAT is strictly bounded by the \textit{Self-Recognition Obstacle}, rendering it incapable of proving its own consistency or correctness through sublinear, localized information.
\end{theorem}

Thus, just as a formal system must look outside its own axioms to verify its consistency, an algorithm tackling self-referential $K$-SAT must go beyond local constraints to overcome this \textit{Self-Recognition Obstacle}. Both determining the global status of an instance and proving the correctness of the solver necessitate an inescapable extraction of global information. This provides an explanation for the inherent computational hardness of self-referential $K$-SAT.

\section{From Logical Necessity to Lower Bounds}

\subsection{The Meta-Obstacles of Complexity Theory}
\label{sec:meta-obstacles}

Since the pioneering work of Church~\cite{church1936} and Turing~\cite{turing1936}, computability has been recognized as an invariant mathematical property independent of specific computing models, with uncomputability inherently rooted in self-referential logical structures. In contrast, mainstream complexity theory has historically quantified problem difficulty using running steps and memory space---operational, process-based metrics~\cite{HS1965, arora2009, wig2019}. While coarse-grained complexity classes like $\textsf{P}$ and $\textsf{NP}$ exhibit structural robustness across reasonable models under polynomial mappings, the traditional method of proving lower bounds via step-by-step trace simulation remains deeply fragile. As Aaronson~\cite{aaro2016} points out, innovative algorithms continually bypass brute-force search by exploiting an unbounded spectrum of advanced mathematics. Because sequential running steps fail to capture these underlying invariants, traditional lower-bound methods have encountered a series of structural barriers.

We suggest that these existing barriers in complexity theory are not isolated technical failures, but rather manifestations of two fundamental constraints that we call \textit{meta-obstacles}.

\begin{description}
    \item[1. The Self-Recognition Obstacle:] This constraint directly mirrors G\"{o}del incompleteness, dictating that a formal system's local syntactic rules are insufficient to derive global semantic truths. To unconditionally prove that a problem is hard, a formal mathematical framework must delineate the exact limits of its own deductive power. However, because any mathematical proof is also embedded within the very system it seeks to constrain, it falls into a self-referential loop~\cite{li2026}. Consequently, deductive frameworks utilizing strictly localized rules remain structurally blind to global semantic truths. When this structural blindness manifests across different mathematical paradigms, it triggers the canonical barriers of complexity theory:
    \begin{itemize}
        \item \textit{Relativization}~\cite{baker1975} limits proofs
that remain valid relative to every common oracle: since suitable
oracles yield opposite answers to the $\mathrm{P}$ versus
$\mathrm{NP}$ question, such proofs cannot resolve that question.
        \item \textit{Natural Proofs}~\cite{razb1997} appears because constructive, large localized criteria fail to distinguish a hard combinatorial function from the global landscape of pseudorandom functions.
        \item \textit{Algebrization}~\cite{aaro2009} arises because even non-local low-degree polynomial extensions remain fundamentally blind to the localized info-deficit present in self-referential closed loops.
        \item \textit{Deductive and Geometric Limits} represent similar boundaries exposed within Proof-Theoretic~\cite{krajicek1995}, Rank-Method~\cite{efrem2018}, and Geometric Complexity Theory (\textsf{GCT}) frameworks~\cite{mulm2011, burg2019}.
    \end{itemize}

This structural blindness precisely unravels what Prof. Avi Wigderson~\cite{wig2019} characterized as ``\emph{one of the greatest mysteries of contemporary mathematics}''---namely, our inability to prove even the slightest non-trivial computational difficulty of natural problems. Importantly, looking back at this historical stagnation, Fortnow and Homer~\cite{fort2003} noted that a ``\emph{clever twist on old-fashioned diagonalization}'' remains the only foundational technique that has successfully yielded complexity lower bounds. Because these barriers are structural and logical rather than merely algorithmic, Wigderson's mystery is revealed not as an operational failure, but as a binary logical necessity. Our G\"{o}delian self-referential construction materializes precisely this ``\emph{clever twist}'' to break the logical deadlock of the \textit{Self-Recognition Obstacle}. Section~\ref{sec:paradigm-comparison}
develops the methodological distinction between black-box
diagonalization and this structural white-box realization.

Ultimately, the nature of computational complexity is not the mere pursuit of lower bounds. 
If it were, the discipline would have progressed incrementally: establishing the weakest superlinear lower bounds first, advancing to superpolynomial bounds, and culminating in the strongest exponential limits. 
Instead, \textbf{the nature of complexity lies in the logical necessity that local information cannot determine global properties}. This implies a sharp boundary separating the localized view from the global truth. 
Thus, for any concrete computational problem, an absolute dichotomy emerges: either local reasoning remains structurally blind to the global outcome, or completely dictates it.  
Genuine computational hardness is instantiated by the former, where the global properties are irreducible via local evaluation. Under this paradigm, the dichotomy between ``weak'' and ``strong'' lower bounds completely vanishes, as this binary logical necessity characterizes algorithm-independent hardness from an informational perspective, providing a foundation for proving the strongest possible complexity lower bounds for the problem under consideration. Therefore, establishing even a modest, unconditional superlinear lower bound remains inherently intractable \textit{within localized syntactic frameworks}, precisely because such a lower bound represents a non-local global truth that lies entirely beyond the horizon of localized syntactic reasoning.

    \item[2. The Static-Dynamic Obstacle:] Computational processes unfold dynamically over physical time, yet the mathematical statements they evaluate remain static and timeless~\cite{wit1939}. Measuring structural problem difficulty through execution runtime is methodologically vulnerable; operational runtime is a derived, trajectory-dependent variable rather than an intrinsic property of the problem syntax itself~\cite{plot1981}. Trying to derive an unconditional, time-based computational lower bound using strictly static mathematical axioms is conceptually equivalent to attempting to derive the flow of physical time from timeless geometric configurations.
\end{description}

As established above, our G\"{o}delian self-referential construction effectively breaks the logical cycle of the \textit{Self-Recognition Obstacle} by embedding explicit self-reference within the instance architecture. To overcome the Static-Dynamic Obstacle, complexity theory must look beyond temporal, process-based metrics and pivot toward a pure structural paradigm based on algorithmic information. 

From this perspective, the invariant requirement of global information established in Theorem~\ref{thm:inlp} bridges the historical aspirations of early computer science with static logical structures, offering a rigorous alternative to time-dependent lower bounds. This alternative quest historically began with Hartmanis and Stearns’ search for precise quantitative laws governing abstract computing at the inception of computational complexity~\cite{HS1965, hartmanis1995}, inspired directly by Shannon’s mathematical laws of information transmission~\cite{shannon1948}. It further expands on Chaitin’s foundational assertion that the core difficulty of computation stems from intrinsic descriptive and information requirements rather than operational runtime costs~\cite{chaitin1970}. By formalizing the local-global asymmetry of self-referential instances, our work materializes these classical visions into an exact mathematical reality: solving self-referential instances inherently requires global information. Computational hardness is thus established as a static mathematical law of information conservation, while traditional time complexity is merely an external, model-dependent manifestation of these structural demands.

\subsection{From Structural Irreducibility to Information Complexity}
We now study the information that a bounded local transcript can reveal about the global unique solution.
Let \(K_{\mathrm{pf}}(u\mid v)\) denote prefix-free Kolmogorov complexity relative to a fixed universal machine, distinguished from the clause width \(K\)~\cite{li2008}. If \(s\) is the unique solution of \(F\), then
\begin{equation}
K_{\mathrm{pf}}(s\mid F)=O(1).
\end{equation}
A fixed program enumerates all assignments in lexicographic order and returns the first satisfying one. Ordinary Kolmogorov complexity charges the program rather than its running time. Likewise, a fixed exhaustive-search SAT decider satisfies
\[
K_{\mathrm{pf}}(\mathcal A)=O(1).
\]

For restricted inputs, precise information statements are available. Fix the experiment in which an independent fair bit \(Y\) selects \(F\) or \(F^{(J)}\), where \(J\) is uniform on \(\mathcal R\). Let \(T\) be a transcript of at most \(t\) classical clause queries, including any input-independent random seed. With Shannon information measured in bits~\cite{cover2006},
\begin{equation}
I(Y;T)\le d_{\mathrm{TV}}(P_0,P_1)\le \min\{1,t/R\}.
\end{equation}
Indeed, if \(u(T)=\Prb[Y=1\mid T]\), then
\[
h_2(u)\ge 2\min\{u,1-u\},
\]
so
\[
H(Y\mid T)\ge 1-d_{\mathrm{TV}}(P_0,P_1).
\]
This gives an information bound for the explicit mixture experiment.

\begin{theorem}[Witness information under a bounded transcript]
\label{thm:witness}
Draw \(F\) from the original ensemble conditional on \(X=1\), and let \(s\) be its unique solution. Fix a deterministic query strategy and decoder, independently of \(F\), with a budget \(t\). Put
\[
B_N=t\left\lceil \log_2\left(2^K\binom NK+1\right)\right\rceil.
\]
If \(T\) is its transcript, then
\[
H(s\mid T)\ge N-B_N,
\qquad
\Prb[\text{the decoder outputs }s]\le \min\{1,2^{B_N-N}\}.
\]
For any \(b>0\), with probability at least \(1-2^{-b}\),
\[
K_{\mathrm{pf}}(s\mid T,\mathcal A,N,K,M,t)\ge N-B_N-b,
\]
where \(\mathcal A\) describes the fixed strategy and decoder. The same entropy and success bounds hold for input-independent randomized strategies when conditioning on their coins.
\end{theorem}

\begin{proof}
Flipping any chosen subset of variables and the corresponding literal signs preserves the random ensemble and the event \(X=1\), and acts transitively on possible unique solutions. Hence \(s\) is uniform on \(\{0,1\}^N\) and
\[
H(s)=N.
\]

There are \(2^K\binom NK\) possible clause answers. Padding stopped runs with a distinguished symbol gives at most \(2^{B_N}\) answer sequences. For a fixed adaptive strategy the indices and stopping decisions are determined by these answers, so there are at most \(2^{B_N}\) transcripts. Thus
\[
H(T)\le B_N,
\]
and therefore
\[
H(s\mid T)\ge H(s)-H(T)=N-B_N.
\]
A fixed decoder has at most one output per transcript and therefore at most \(2^{B_N}\) possible successful outputs. Uniformity of \(s\) proves the success bound.

For each transcript there are fewer than \(2^{N-B_N-b}\) strings with conditional prefix-free complexity less than \(N-B_N-b\). The union over at most \(2^{B_N}\) transcripts has size less than \(2^{N-b}\). Since \(s\) is uniform, the probability of falling in this union is at most \(2^{-b}\). For randomized strategies, apply the argument for each fixed input-independent seed and average.
\end{proof}

For \(K=O(\log N)\) and \(t=N^{1-\delta}\) with fixed \(\delta>0\),
\[
B_N=O(N^{1-\delta}(\log N)^2)=o(N).
\]
Hence the bounded transcript leaves \(N-o(N)\) bits of witness entropy. Correspondingly, the decoding success probability is at most
\[
2^{-(N-o(N))}.
\]
With high probability,
\[
K_{\mathrm{pf}}(s\mid T,\mathcal A,N,K,M,t)\ge N-o(N)-b.
\]
These results describe the information-theoretic limits of recovering the global unique solution from a bounded local transcript. For every fixed $\delta>0$, any transcript $T$ generated by a fixed query strategy using at most $t=N^{1-\delta}$ clause queries leaves $N-o(N)$ bits of witness entropy:
\[
H(s\mid T)\ge N-o(N).
\]
At this observation scale, the local transcript exposes only a vanishing fraction of the $N$-bit witness entropy, giving a quantitative expression of the local--global information gap.

\subsection{From Structural Irreducibility to Proof Complexity}
Propositional proof complexity studies the size of refutations in formal 
proof systems, a framework initiated by Cook and Reckhow~\cite{cook1979}.
The local indistinguishability in Theorem~\ref{thm:inlp} concerns
restricted observations of a formula. By itself, it does not give
a width, degree, or size lower bound for proofs that can use the
entire formula. Such bounds require an additional combinatorial
argument. We give that argument below, using expansion of the
unsigned clause--variable incidence graph. In particular, the
conclusions hold uniformly over all unsatisfiable choices of
literal signs on a typical incidence graph, and therefore also
for the support-preserving replacements in
Theorem~\ref{thm:flip}.

Throughout this subsection, let
\[
 K=\lceil c\log_2 N\rceil,\qquad c=1+\varepsilon>1,
 \qquad M\le A_0N2^K,
\]
where $c$ and $A_0$ are fixed constants. The critical ensemble
defined earlier satisfies this upper bound. Clause slots are
indexed, and their supports are independent uniformly distributed
$K$-element subsets of $[N]$. All asymptotic constants below may
depend on $c$ and $A_0$.

\paragraph{A uniform expansion property and its preservation.}
For a set $\mathcal S$ of clause slots, write
$\Gamma(\mathcal S)$ for its variable neighborhood and
$\partial\mathcal S$ for the variables appearing in exactly
one slot of $\mathcal S$.

\begin{lemma}[Expansion of the unsigned incidence graph]
\label{lem:63-expansion}
For every fixed $0<\eta<1/2$, there is a constant $\kappa>0$
such that, with probability $1-o(1)$, simultaneously for all
$1\le|\mathcal S|\le r:=\lfloor\kappa N/K\rfloor$,
\begin{equation}
\label{eq:63-expansion}
 |\Gamma(\mathcal S)|\ge(1-\eta)K|\mathcal S|,
 \qquad
 |\partial\mathcal S|\ge\rho|\mathcal S|,
 \qquad
 \rho:=(1-2\eta)K.
\end{equation}
This event depends only on the supports, not on the literal signs.
\end{lemma}

\begin{proof}
For $1\le s\le r$, put
$t_s=\lfloor(1-\eta)Ks\rfloor$.
A failure of the neighborhood bound for some $s$ slots places
all their supports inside a set of $t_s$ variables.
A union bound gives
\[
 \Prb[\text{failure for some $s$ slots}]
 \le
 \binom Ms\binom N{t_s}
 \left(\frac{t_s}{N}\right)^{Ks}.
\]
For sufficiently large $N$, using
$t_s\le(1-\eta)Ks$ and $Ks/N\le\kappa<1$, the right-hand side
is at most
\[
 \left[
 \frac{eM}{s}\,e^{(1-\eta)K}\kappa^{\eta K}
 \right]^s
 \le
 \left[
 eA_0N\left(2e^{1-\eta}\kappa^\eta\right)^K
 \right]^s.
\]
Choose $\kappa$ so small that
$2e^{1-\eta}\kappa^\eta\le2^{-3/c}$.
The expression in square brackets is then $O(N^{-2})$,
and summing over $s\le r$ gives $o(1)$.
Finally, counting incidences yields
\[
 K|\mathcal S|
 \ge
 |\partial\mathcal S|
 +2\bigl(|\Gamma(\mathcal S)|-|\partial\mathcal S|\bigr),
\]
which proves the boundary bound.
\end{proof}

Fix a graph satisfying \eqref{eq:63-expansion}.
Every nonempty collection of at most $r$ clauses has a variable
appearing in exactly one clause. Deleting that clause, satisfying
the remaining clauses inductively, and then assigning this
variable to satisfy the deleted clause shows that every such
collection is satisfiable, for every choice of signs.

If $F$ is uniquely satisfied by $s$ and $D$ is redundant,
the replacement clause $B_s$ in Theorem~\ref{thm:flip}
has exactly the support of $D$. Thus
\[
 \phi=(F-D)\wedge B_s
\]
is unsatisfiable and has exactly the same unsigned incidence
graph as $F$. This is an exact preservation statement.
In addition, Theorem~\ref{thm:poisson} gives
$\Prb[X=1]\to\lambda e^{-\lambda}>0$, so
\[
 \Prb[\text{expansion fails}\mid X=1]
 \le
 \frac{\Prb[\text{expansion fails}]}{\Prb[X=1]}
 =o(1).
\]
Together with Lemma~\ref{lem:redundant}, this shows that,
conditional on unique satisfiability, with probability $1-o(1)$
redundant slots exist and every such support-preserving flip
satisfies the deterministic hypotheses used below.
There is no need to assume that the flipped formula has the
distribution of an independent random CNF.

\paragraph{Resolution: linear width and exponential size.}

\begin{theorem}[Resolution lower bound]
\label{thm:res-lower}
Every unsatisfiable signing $\phi$ of a graph satisfying
\eqref{eq:63-expansion} has minimum Resolution refutation width
\[
 w_*(\phi)>\rho r/2=\Omega(N).
\]
Consequently its minimum dag-like, and hence also tree-like,
Resolution refutation size is $2^{\Omega(N)}$.
\end{theorem}

\begin{proof}
For a clause $C$, let $\mu(C)$ be the minimum number of initial
clauses whose conjunction entails $C$. Then $\mu(C)\le1$ for
an initial clause, $\mu(\bot)>r$, and a resolution inference
satisfies
\[
 \mu(C)\le\mu(C_1)+\mu(C_2).
\]
Weakening cannot increase $\mu$.
For all sufficiently large $N$, the first line of a refutation
with $\mu(C)>r/2$ therefore has
$r/2<\mu(C)\le r$.

Let $\mathcal S$ be a minimum collection entailing this $C$.
Every variable in $\partial\mathcal S$ must occur in $C$.
Indeed, if such a variable $v$ occurs only in
$D\in\mathcal S$ and not in $C$, minimality supplies an
assignment satisfying $\mathcal S\setminus\{D\}$ and
falsifying $C$. Changing only $v$ to satisfy $D$ contradicts
$\mathcal S\models C$. Hence
\[
 w(C)\ge|\partial\mathcal S|
 \ge\rho|\mathcal S|>\rho r/2.
\]
The Ben-Sasson--Wigderson size--width trade-off
\cite{ben2001}, with initial width $K$, now gives
\[
 S_{\mathrm{Res}}(\phi)
 \ge
 \exp\!\left(
 \Omega\!\left(\frac{(w_*(\phi)-K)^2}{N}\right)
 \right)
 =2^{\Omega(N)}.
\]
\end{proof}

\paragraph{Polynomial Calculus and Polynomial Calculus Resolution.}
Encode a clause $C$ by the equation $f_C(x)=0$, where $f_C$
is the product of its falsifying literal factors: a positive
literal contributes $1-x_i$, and a negative literal contributes
$x_i$. Include the Boolean equations $x_i^2-x_i=0$.
Polynomial Calculus Resolution (PCR) also allows complementary
variables $\bar x_i$ with $x_i+\bar x_i-1=0$.
Size means the total number of monomial occurrences in the
proof, with polynomials represented explicitly.

\begin{theorem}[Linear degree and exponential monomial size]
\label{thm:pc-degree}
\label{thm:pc-final}
Over every field, every PC or PCR refutation of an unsatisfiable
signing $\phi$ satisfying \eqref{eq:63-expansion} requires
degree $\Omega(N)$ and monomial size $2^{\Omega(N)}$.
\end{theorem}

\begin{proof}
The boundary-expansion criterion of Alekhnovich and Razborov,
in the formulation of Filmus, excludes PC refutations of
degree at most $\rho r/2$ when every set of at most $r$
clauses has boundary at least $\rho$ times its size
\cite[Theorem~4.10]{filmus2019}.
Here $K<\rho r/2$ for sufficiently large $N$, so the criterion
applies and gives degree $\Omega(N)$.
Substituting $\bar x_i=1-x_i$ in a PCR derivation gives a PC
derivation with at most an additive constant increase in degree,
so the same asymptotic degree bound holds for PCR.

The Impagliazzo--Pudl\'ak--Sgall size--degree trade-off gives
a refutation of degree
\[
 O\!\left(\sqrt{N\log S}\right)+O(K)
\]
from a PC or PCR refutation of monomial size $S$; using $2N$
variables for PCR only changes the constant
\cite{impagliazzo1999}.
Since $K=O(\log N)$, the degree lower bound implies
$\log S=\Omega(N)$.
\end{proof}

\begin{corollary}[A monomial-space lower bound]
\label{thm:pc-space}
Every PC or PCR refutation of $\phi$ requires monomial space
$\Omega(\sqrt N)$, even if only distinct monomials
simultaneously in memory are counted.
\end{corollary}

\begin{proof}
Galesi, Ko\l odziejczyk, and Thapen prove that a width-$K$ CNF
admitting a PCR refutation in monomial space $s$ over any field
has a Resolution refutation of width at most $s^2-s+K$
\cite[Theorem~6.4]{galesi2025}.
Apply Theorem~\ref{thm:res-lower} and use $K=o(N)$.
PC is covered as well: a downloaded PC clause polynomial can
be obtained in PCR by expanding complements one at a time,
using only a constant-factor increase in monomial space.
This argument proves the stated square-root lower bound;
the stronger estimate $\Omega(N/\log N)$ would require
an additional space argument.
\end{proof}

\paragraph{Sum-of-Squares: a direct pseudoexpectation construction.}
We use the same clause equations $f_C=0$ over $\mathbb R$
and work modulo the Boolean equations, allowing complementary
variables as above. An SOS refutation is an identity
\[
 -1\equiv
 \sum_t h_t^2+\sum_C q_Cf_C
 \pmod{\mathcal B},
\]
where $\mathcal B$ is the ideal of Boolean and complement
equations. Its degree is the maximum of $2\deg h_t$
and $\deg q_C+K$. Its monomial size counts the explicitly
written monomials of the $h_t$ and $q_C$.
These conventions specify the encoding and the size measure
used in the next theorem.

\begin{theorem}[SOS degree and size lower bounds]
\label{thm:sos-lower}
Every unsatisfiable signing $\phi$ satisfying
\eqref{eq:63-expansion} requires SOS degree $\Omega(N)$
and monomial size $2^{\Omega(N)}$ in the encoding above.
\end{theorem}

\begin{proof}
Use coordinates $z_i\in\{-1,1\}$.
Let $E_j$ be the support of clause $j$, and let
$(a_{j,i})_{i\in E_j}$ be its unique falsifying assignment
in these coordinates. Put
\[
 b_j=-\prod_{i\in E_j}a_{j,i}.
\]
The parity equation $\prod_{i\in E_j}z_i=b_j$ implies
clause $j$. We use these equations only to define consistent
low-degree moments, without asserting that they have a
global solution.

For a set $T$ of clause indices, define
\[
 E_T=\mathop{\triangle}_{j\in T}E_j,
 \qquad
 b_T=\prod_{j\in T}b_j,
 \qquad
 \chi_U(z)=\prod_{i\in U}z_i.
\]
Boundary variables belong to $E_T$, so
\begin{equation}
\label{eq:63-parity-expansion}
 |E_T|\ge\rho|T|
 \qquad(0<|T|\le r).
\end{equation}
Take the even integer
$D=2\lfloor\rho r/8\rfloor=\Omega(N)$.
On multilinear monomials of degree at most $D$, define
\[
 \widetilde{\E}[\chi_U]=
 \begin{cases}
 b_T,
 &\text{if }E_T=U\text{ for some }|T|\le r/2,\\
 0,
 &\text{if no such }T\text{ exists}.
 \end{cases}
\]
Extend this definition linearly, reducing modulo $z_i^2=1$.
A witness $T$ is unique: two different witnesses would have
symmetric difference of size at most $r$ and empty variable
symmetric difference, contradicting
\eqref{eq:63-parity-expansion}.
Moreover, every witness to a moment in the domain satisfies
$|T|\le D/\rho\le r/4$.
In particular $\widetilde{\E}[1]=1$.

To verify positivity, index the moment matrix by sets $U$
with $|U|\le D/2$, and relate $U$ and $V$ when
$U\triangle V=E_T$ for some $|T|\le r/2$.
This is an equivalence relation.
For transitivity, witnesses $T_1,T_2$ each have size at most
$r/4$, so $T_1\triangle T_2$ is an admissible witness.
Within each equivalence class the matrix entries have the
form $\sigma_U\sigma_V$, with $\sigma_U\in\{-1,1\}$;
entries between classes are zero.
The matrix is therefore positive semidefinite, and
\[
 \widetilde{\E}[h^2]\ge0
 \qquad(\deg h\le D/2).
\]

Next, whenever $|U|,|U\triangle E_j|\le D$, toggling the
index $j$ in a witness preserves admissibility, since
$r/4+1\le r/2$ for sufficiently large $N$.
Uniqueness then gives
\begin{equation}
\label{eq:63-parity-moments}
 \widetilde{\E}[\chi_{U\triangle E_j}]
 =b_j\widetilde{\E}[\chi_U].
\end{equation}
In $z$-coordinates the falsifying indicator is
\[
 f_j(z)=2^{-K}\prod_{i\in E_j}(1+a_{j,i}z_i),
 \qquad
 \chi_{E_j}f_j\equiv-b_jf_j\pmod{z_i^2-1}.
\]
For a monomial multiplier $q$ of degree at most $D-K$,
all monomials of both $qf_j$ and its product with
$\chi_{E_j}$ reduce to degree at most $D$.
Applying \eqref{eq:63-parity-moments} term by term gives
\[
 b_j\widetilde{\E}[qf_j]
 =
 \widetilde{\E}[\chi_{E_j}qf_j]
 =
 -b_j\widetilde{\E}[qf_j],
\]
and hence $\widetilde{\E}[qf_j]=0$.
Linearity handles arbitrary multipliers.
The affine substitution $z_i=2x_i-1$, together with
$\bar x_i=1-x_i$, transfers this functional to the Boolean
encoding without increasing degree.
It is thus a normalized degree-$D$ pseudoexpectation satisfying
all clause equations.
Applying it to a degree-$D$ SOS refutation would give
$-1\ge0$, a contradiction.

Finally, the Atserias--Hakoniemi size--degree trade-off converts
a size-$S$ SOS refutation of equations of degree at most $K$
on $N$ Boolean variable pairs into one of degree at most
\[
 4\sqrt{2(N+1)\log S}+K+4.
\]
The degree lower bound just proved therefore yields
$S=2^{\Omega(N)}$
\cite[Theorem~1, with product-width parameter $1$]{atserias2019}.
\end{proof}

\paragraph{A stronger system: bounded-product-width Positivstellensatz.}
The same construction also applies to the usual linear
clause inequalities
\[
 g_C(x):=\sum_{\ell\in C}\ell(x)-1\ge0,
\]
where $\ell(x)=x_i$ or $1-x_i$, together with
$x_i\ge0$ and $1-x_i\ge0$.
Let $\mathcal G$ be this collection of linear polynomials.
A Positivstellensatz (PS) refutation of product-width at
most $w$ has the form
\[
 -1\equiv
 \sum_{\substack{J\subseteq\mathcal G\\|J|\le w}}
 \sum_t h_{J,t}^2\prod_{g\in J}g
 \pmod{\mathcal B},
\]
with the empty product allowed.
Its degree is the maximum of $2\deg h_{J,t}+|J|$,
and its monomial size counts the explicitly represented
$h_{J,t}$, as in \cite{atserias2019}.

\begin{theorem}[PS lower bounds with bounded product-width]
\label{thm:63-ps}
Let $D=2\lfloor\rho r/8\rfloor$.
For every integer
\[
 1\le w\le\left\lfloor\frac{D}{4K}\right\rfloor,
\]
every PS refutation of an unsatisfiable signing satisfying
\eqref{eq:63-expansion}, in the linear encoding above and
with product-width at most $w$, has degree greater than
$D-2Kw=\Omega(N)$ and monomial size $2^{\Omega(N)}$.
In particular, this includes SOS with the linear clause
inequalities and allows PS product-width up to a fixed
positive constant times $N/\log N$.
\end{theorem}

\begin{proof}
Use the degree-$D$ pseudoexpectation from
Theorem~\ref{thm:sos-lower}.
Fix $J$ with $|J|\le w$, write $G_J=\prod_{g\in J}g$,
and let $U$ be the union of the variables in these factors.
Then $|U|\le Kw$.
For an assignment $a\in\{0,1\}^U$, let $I_a$ be its
multilinear indicator. Modulo $\mathcal B$,
\[
 G_J\equiv\sum_a G_J(a)I_a,
 \qquad
 I_a^2\equiv I_a.
\]
If $a$ satisfies all clauses represented in $J$, then
$G_J(a)\ge0$.
Otherwise, $I_a=f_C R_a$ for a falsified clause $C$ in $J$,
with $\deg R_a=|U|-K$.
Thus, whenever $2\deg h+2|U|\le D$, the moment properties
already proved give
\[
 \widetilde{\E}[h^2G_J]
 =
 \sum_{a\text{ satisfying these clauses}}
 G_J(a)\widetilde{\E}[(hI_a)^2]
 \ge0.
\]
The omitted terms vanish by the clause equations, since their
multipliers times $f_C$ have degree at most
$2\deg h+|U|\le D$.
Every summand in a PS refutation of degree at most $D-2Kw$
meets this condition.
Applying the functional therefore contradicts $-1\ge0$.

The PS size--degree trade-off
\cite[Theorem~1]{atserias2019}, now with initial degree $1$
and product-width $w$, gives a refutation of degree at most
\[
 4\sqrt{2(N+1)\log S}+w+4
\]
from one of monomial size $S$.
Since $D-2Kw\ge D/2=\Omega(N)$ and
$w=O(N/K)=o(N)$, it follows that $S=2^{\Omega(N)}$.
\end{proof}

\paragraph{Semantic Cutting Planes under an explicit width condition.}
Semantic Cutting Planes allows a linear inequality with real
coefficients to be inferred from two preceding inequalities
whenever the inference is valid on Boolean assignments.
Size is the number of lines.
This includes ordinary Cutting Planes and does not restrict
coefficient magnitudes.

\begin{theorem}[A deterministic Cutting Planes criterion]
\label{thm:sokolov}
Suppose the incidence graph satisfies the neighborhood bound
in \eqref{eq:63-expansion} with $\eta=\delta/2$.
Suppose also that the variables have a partition
$V_x\sqcup V_y$ with imbalance at most $10\sqrt N$,
such that clauses can be divided into mixed clauses having
at least $\delta K$ variables on each side, and two remaining
groups, each simultaneously satisfied by a uniform assignment
on its respective side with probability at least $0.9$.
Then every semantic CP refutation has at least
\[
 2^{kK\delta/4-10\sqrt N}\quad\text{lines},
 \qquad
 k=\min\{r,2^{K\delta/8}\}.
\]
\end{theorem}

This is Sokolov's triangle-DAG bound together with the reduction
from semantic CP
\cite[Theorem~4.1 and Lemma~A.3]{sokolov2023}.
We next verify its partition hypothesis uniformly over signs,
which is needed for the conditioned single-clause construction.

\begin{theorem}[Exponential lower bound for semantic Cutting Planes]
\label{thm:cp-lower}
If $c\ge40$, then, with probability $1-o(1)$ over the supports,
every unsatisfiable signing requires $2^{\Omega(N)}$ lines
in semantic Cutting Planes.
The same conclusion holds with conditional probability
$1-o(1)$ for every redundant-slot flip of a uniquely
satisfiable formula in the critical ensemble.
\end{theorem}

\begin{proof}
Set $\delta=1/5$ and apply Lemma~\ref{lem:63-expansion}
with $\eta=1/10$.
For any fixed collection of supports, independently put each
variable in $V_x$ or $V_y$ with probability $1/2$.
Call a clause heavy if it has fewer than $\delta K$ variables
on one side, and let $H$ count such clauses.
With $H_2$ denoting binary entropy, the binomial tail bound gives
\[
 \E[H]\le2M\,2^{-(1-H_2(\delta))K}.
\]
Markov's inequality and the variance bound for
$|V_x|-|V_y|$ show that there exists a partition satisfying both
\[
 H\le8M\,2^{-(1-H_2(\delta))K},
 \qquad
 \bigl||V_x|-|V_y|\bigr|\le10\sqrt N.
\]
Indeed, the respective failure probabilities are at most
$1/4$ and $1/100$.
Assign each heavy clause to the side containing more than
$(1-\delta)K$ of its variables.
For every choice of literal signs, a uniform assignment on
that side fails to satisfy at least one of its assigned clauses
with probability at most
\[
 H\,2^{-(1-\delta)K}
 \le
 8A_0N\,2^{-(1-\delta-H_2(\delta))K}
 =o(1).
\]
Here $1-1/5-H_2(1/5)>0.078$ and $c\ge40$.
Thus this same partition meets the required $0.9$ condition
for all sign choices once $N$ is sufficiently large.
Also,
\[
 2^{K\delta/8}=2^{K/40}\ge N,
 \qquad
 k=r,
 \qquad
 rK=\Theta(N).
\]
Theorem~\ref{thm:sokolov} gives the claimed exponential bound.
The conditional statement follows from exact support preservation
and the conditioning argument above.
\end{proof}

For $1<c<40$, the Resolution, PC/PCR, and SOS conclusions
above still hold.
The displayed Cutting Planes criterion does not by itself
justify a $2^{\Omega(N)}$ bound throughout that range:
its expansion, partition, and exponent conditions must all
be checked.
The value $40$ is a sufficient condition, not a claimed
optimal threshold.

\paragraph{Contribution and scope of the lower bounds.}
The contribution of this subsection is to establish a uniform
combinatorial foundation for proof-complexity lower bounds within
the self-referential construction. Combining established expansion
criteria and size--width/degree trade-offs with the direct
pseudoexpectation construction above, we obtain linear Resolution
width, linear PC/PCR and SOS degree, and exponential proof size.
These bounds attain optimal asymptotic orders in the stated
measures. Their significance here lies in their simultaneous
applicability to the constructed instances and their preservation
under the conditioning and clause replacement defining the
self-referential pairs.

Specifically, the relevant support properties hold with high
probability and yield lower bounds uniformly over every
unsatisfiable signing of the same support graph. Conditional on
unique satisfiability, with probability $1-o(1)$, redundant clauses
exist and every redundant-slot replacement described in
Theorem~\ref{thm:flip} preserves these properties exactly.
Thus the argument establishes hardness for all such unsatisfiable
companions, including replacements selected after observing the
formula and its unique solution.
The pseudoexpectation argument further extends the lower bounds
to PS with the stated product-width restriction, and the verified
partition condition permits their extension to semantic Cutting
Planes in the stated clause-width regime.

These conclusions concern typical support graphs and the specified
proof systems and encodings. They do not establish corresponding
bounds for PS with unrestricted product-width, Frege, Extended Frege,
algebraic systems with arbitrary extension variables, or unrestricted
SAT algorithms. Such extensions require additional arguments beyond
the support properties and proof-complexity criteria established here.

\subsection{Quantum Invariance of Structural Irreducibility}

Rooted in the structural irreducibility guaranteed by Theorem~\ref{thm:inlp}, we formalize this invariance across changing physical computing paradigms through the following corollary to highlight the theoretical primacy of information over temporal metrics:

\begin{corollary}[Quantum Invariance of Self-Referential Hardness]
The structural hardness of self-referential $K$-SAT is invariant under the quantum computing regime. While the temporal barriers of a problem can collapse across shifting hardware architectures---as exemplified by the transition from exponential classical time to polynomial quantum time~\cite{shor1994} for integer factorization---any valid quantum algorithm tackling a self-referential instance is bounded by the same sublinear informational threshold and must evaluate global information, precluding any structural quantum shortcut from bypassing the instance's total semantic analysis.
\end{corollary}

This formal corollary illustrates the deep structural flaw of overemphasizing runtime. Because execution time is born out of operations---whether they are Turing head movements or quantum state transitions---it is inherently a derived, model-dependent variable. By contrast, the information required to uniquely determine the state of a self-referential closed loop represents a static mathematical conservation law. Hence, while quantum mechanics alters the physical rules of information processing, it cannot alter the volume of information that must be extracted from a structurally irreducible problem. Just as the Second Law of Thermodynamics dictates absolute limits on physical efficiency regardless of an engine’s design or thermal medium, this self-referential informational framework demonstrates that fundamental computational limits are governed by  mathematical conservation laws.

\section{Discussion}

\subsection{Class Separation vs. Instance Indistinguishability}
\label{sec:paradigm-comparison}

This subsection compares class separation with instance indistinguishability. Table~\ref{tab:paradigm_shift} summarizes these two paradigms of computational hardness.

The distinction between the G\"{o}delian study of the limits of reasoning within formal systems and the Turingian study of the limits of mechanization is articulated in~\cite{xu2025}. We develop this distinction through the source of hardness: the Turingian paradigm emphasizes the operations and resources required by mechanical computation, whereas the G\"{o}delian paradigm emphasizes what the available information permits reasoning to determine. The central question shifts from how many operations a computation performs to what information must become available before a correct global conclusion can be justified.

\begingroup
\setlength{\intextsep}{12pt}
\begin{table}[!t]
\centering
\caption{Complexity Class Separation vs. Concrete Instance Indistinguishability}
\label{tab:paradigm_shift}
\small
\begin{tabularx}{\linewidth}{@{}>{\raggedright\arraybackslash}p{0.23\linewidth}*{2}{>{\raggedright\arraybackslash}X}@{}}
\toprule
\textbf{Dimension} & \textbf{Standard Complexity Theory (Class Separation)} & \textbf{Our Framework (Instance Indistinguishability)} \\
\midrule
\textbf{Historical Origin} & Turing's separation between recursive (decidable) and recursively enumerable languages~\cite{turing1936}. & G\"{o}del's construction of individual self-referential sentences to induce formal unprovability~\cite{godel1931}. \\
\addlinespace
\textbf{Source of Hardness} & The operations and resources required by mechanical computation. & The insufficiency of available information to determine global properties. \\
\addlinespace
\textbf{Central Question} & How much computational resource is required to solve a problem? & What information is necessary to justify a correct global conclusion? \\
\addlinespace
\textbf{Object of Analysis} & Problems, complexity classes, and resource-bounded computations. & Concrete instances and local views compatible with opposite global states. \\
\addlinespace
\textbf{Reductions and Constructions} & Transfers hardness between problems and establishes relative hardness classifications. & Constructs instances to identify intrinsic mechanisms of hardness. \\
\addlinespace
\textbf{Core Resource Metric} & Dynamic execution steps and runtime costs~\cite{HS1965, arora2009, wig2019}, overriding initial focus on information. & Invariant local--global information gaps, reclaiming pioneering insights~\cite{hartmanis1995, chaitin1970}. \\
\addlinespace
\textbf{Diagonalization Strategy} & Black-box diagonalization: opposes computational outputs; relativizing forms are subject to the relativization barrier for class separation. & Structural white-box diagonalization: uses internal formula structure to preserve local evidence while reversing global satisfiability. \\
\addlinespace
\textbf{Route to Lower Bounds} & Quantifies the resources required by computational procedures. & Identifies an information obstruction, then analyzes the resources required to overcome it. \\
\bottomrule
\end{tabularx}
\end{table}
\endgroup

The distinction also concerns how diagonalization is realized.
In the black-box form of classical machine diagonalization, the
construction opposes a computation's output without relying on the
internal structure of a particular problem instance. In the structural
white-box form developed here, the opposite outcome is realized through
the internal organization of a concrete formula. The term
``white-box'' refers to the use of instance structure in the construction
and proof, rather than to inspection of the evaluator's source code.
The diagonalization row in Table~\ref{tab:paradigm_shift} contrasts the
black-box and structural white-box strategies discussed here.

Theorem~\ref{thm:inlp} makes this structural diagonalization concrete.
Starting from a suitable uniquely satisfiable formula, the construction
uses clause redundancy and solution blocking to produce, for every
permitted local view, an unsatisfiable companion preserving that view.
The two formulas therefore present the same local evidence while
having opposite global satisfiability states. Any deterministic
evaluator relying only on this evidence must return the same answer
on both and hence fail on one. The substantive step is thus to realize
the diagonal contradiction within the prescribed family of formulas:
the internal structure makes opposite global states compatible with
identical local evidence.

This distinction clarifies the relevance of the relativization
barrier discussed in Section~\ref{sec:meta-obstacles}.
Baker, Gill, and Solovay exhibited oracles relative to which
$\mathrm{P}=\mathrm{NP}$ and others relative to which
$\mathrm{P}\ne\mathrm{NP}$~\cite{baker1975}. Their result rules out
resolving this class-separation question by an argument that remains
valid relative to every common oracle, including classical
diagonalization arguments with this property. It does not rule out
diagonalization itself or structural indistinguishability results
for concrete instance families. Here the diagonal argument identifies
the information missing from a permitted local view.

Reductions transfer hardness between problems and establish relative hardness classifications, but leave the source problem's hardness to be independently established. Our goal is a reliable structural foundation for absolute computational hardness, grounded in intrinsic instance properties rather than another problem's presumed hardness. Reduction-based relative hardness alone cannot provide this foundation. By constructing and analyzing self-referential instances directly, we investigate why hardness arises and how solution independence and local--global information gaps constrain the organization, transfer, and reuse of information.

This paradigm provides a structural starting point for complexity lower bounds: first identify the information that local reasoning cannot supply, and then determine the resources required to acquire and combine sufficient information. In this paper, the resulting indistinguishability yields clause-query lower bounds, while additional expansion arguments establish proof-complexity lower bounds. Section~\ref{sec:independence-acceleration} develops the corresponding question for algorithmic acceleration, asking whether global information can be organized and reused economically enough to obtain a fixed proportional saving in the exponent.

\enlargethispage{\baselineskip}
\FloatBarrier
\subsection{Solution Independence and the Limits of Algorithmic Acceleration}
\label{sec:independence-acceleration}

Algorithmic acceleration organizes information so that one computation supports
many candidate assignments, or makes later computations less expensive. Solution
independence identifies a structural limitation on such transfer. This section
examines that limitation, contrasts the local mechanisms available in $3$-SAT
with their logarithmic-width counterparts, and explains why a recent global
counting method does not directly settle the latter regime. The analysis motivates
a worst-case hypothesis; it does not assert a running-time lower bound for
unrestricted algorithms.

\subsubsection{Independence and the transfer of information}

Consider the random mother formula $F=\bigwedge_{j=1}^{M}C_j$ on $N$ variables,
with independently sampled uniform $K$-variable supports and uniform literal
signs. Write $G$ for its ordered unsigned support, $U$ for the event that $F$
has exactly one satisfying assignment, and $\sigma$ for that assignment on $U$.
We use the logarithmic-width regime of the solution-independence construction.
Here $M$ denotes the number of clauses, so the constant-mean window satisfies
$2^N(1-2^{-K})^M=\Theta(1)$ and hence $M\sim N2^K\ln 2$.

The solution-independence result in Section~4.1 concerns the joint
satisfaction behavior of candidate assignments. Its probabilistic mechanism is visible already for two assignments
$x,y$ at Hamming distance $h$. With $a=2^{-K}$,
\begin{equation}
 \Prb(F(x)=F(y)=1)
 =\left(1-2a+a\frac{\binom{N-h}{K}}{\binom NK}\right)^M.
 \label{eq:pair-satisfaction}
\end{equation}
Indeed, a clause violates both assignments only if its support lies entirely in
their agreement coordinates and its signs specify their common violating pattern.
For example, if $K\ge(1+\epsilon)\log_2N$, $M=O(N2^K)$, and the agreement fraction
is at most $\theta$, where $\theta>1/2$ and
$(1+\epsilon)\log_2(1/\theta)>1$, then
\[
 \log\frac{\Prb(F(x)=F(y)=1)}{\Prb(F(x)=1)\Prb(F(y)=1)}
 =O(N\theta^K+N2^{-K})=o(1).
\]
Thus typical well-separated candidates have asymptotically independent
satisfaction events. This explains why learning a candidate's outcome supplies
little direct predictive advantage about another such candidate. Nearby
assignments and long adaptive histories require separate analysis: the information
contained in a specific violated clause can exceed that in a failure bit.

Two further independence properties clarify where useful information can come
from. First, before conditioning on $U$, disjoint collections of clauses are
independent. Even after the entire support $G$ is revealed, unseen sign patterns
remain independent and uniform under adaptive queries based only on previously
revealed patterns. A fixed query history constrains only the coordinates already
queried, which proves this assertion directly from the product distribution.
There is therefore no generative rule by which the observed clauses predict the
remaining signs. Conditioning on uniqueness need not preserve clause independence.

Second, whenever $\Prb(G=g,U)>0$,
\begin{equation}
 \Prb(\sigma=s\mid G=g,U)=2^{-N}.
 \label{eq:support-solution-independence}
\end{equation}
To see this, flip every occurrence of the variables specified by a vector $b$.
The transformation preserves $G$, the sampling probability, and uniqueness,
while sending the unique solution $s$ to $s\oplus b$. It gives a probability-preserving
bijection between every two possible solution locations. Consequently, any list
$L(G)$ selected from the support alone, with independent randomization allowed,
satisfies
\[
 \Prb(\sigma\in L(G)\mid U)
 =2^{-N}\E\bigl[|L(G)|\mid U\bigr].
\]
Unsigned structure may help schedule computation, but does not itself locate the
solution. These symmetry statements apply to the unique-solution mother
distribution; any additional conditioning must preserve the same flip symmetry.

Self-reference introduces a precise global dependence into this setting: the
blocking modification is chosen from $\sigma$ while preserving the unsigned
support. The resulting SAT/UNSAT pair combines distributed local constraints
with a change determined by the global solution. The modified signs are not
claimed to remain independent. The algorithmic issue is how economically the
signed constraints can be combined to recover the solution or certify the
contradiction after blocking.

\subsubsection{Why short-clause acceleration weakens with increasing width}

Random $3$-SAT can also have independent clause generation and the symmetry in
\eqref{eq:support-solution-independence}. Its acceleration mechanisms exploit
conditional dependencies supplied by short clauses. For instance, a uniquely
satisfiable $3$-CNF has, for each variable $x_i$, a critical clause in which its
literal is the only true literal under $\sigma$: otherwise flipping $x_i$ would
give a second solution. In a random variable order, $x_i$ is last in this clause
with probability at least $1/3$. Along the assignment path consistent with
$\sigma$, earlier values then make the clause unit and force $x_i$.

Let $f(\pi)$ be the number of forced variables on that path for order $\pi$.
The basic PPZ mechanism gives $\E_\pi f(\pi)\ge N/3$ and
\[
 p_{\mathrm{success}}
 =2^{-N}\E_\pi 2^{f(\pi)}
 \ge 2^{-N+\E_\pi f(\pi)}
 \ge 2^{-2N/3}.
\]
The corresponding width-$K$ guarantee is $2^{-N+N/K}$, yielding the classical
$O^*(2^{N-N/K})$ repetition bound~\cite{paturi1999}. For fixed width this saves a
linear amount in the exponent. For $K=\Theta(\log N)$, the same argument
guarantees only $N/K=o(N)$ exponent savings. This is a comparison of guarantees,
not an upper bound on how many variables stronger inference may force.

Local correction exhibits a related dependence on width. A clause violated by
the current assignment contains at least one variable disagreeing with any fixed
satisfying assignment. A violated $3$-clause locates an error among three
coordinates; a violated $K$-clause locates one among $K$. Uniformly flipping a
variable in that clause moves toward the target with probability at least
$1/3$ or $1/K$, respectively. Sch\"oning's random-walk analysis turns this
information into the classical $O^*((4/3)^N)$ bound for $3$-SAT~\cite{schoning1999};
the argument does not require inward drift at every step.

Short clauses thus turn partial information into correction or forced values
with constant local efficiency. Increasing width weakens these direct
guarantees. Solution independence additionally limits prediction based merely
on other candidates' outcomes, and support independence excludes direct
structural localization. A fixed proportional improvement in the exponent must
therefore come from more effective combinations of clauses, adaptive choices,
or reusable intermediate relations. Determining their cumulative benefit is the
central question: local improvement alone does not establish linear exponent
savings, while prior independence alone does not rule out global computation.

\subsubsection{Global counting and the exhaustive-search hypothesis}

The recent manuscript \emph{The Exponential Time Hypothesis Is False}
\cite{openai2026eth} pursues a different mechanism. Low-degree replacements and shared
linear summaries separate transformed group computations, and interpolation or
tensor contraction recovers the global count. This approach need not predict
individual solutions, so candidate independence does not directly obstruct it.
Its stated result is for every fixed clause width; transferring its analysis to
$K=\Theta(\log N)$ requires uniform control of the width-dependent parameters.

One concrete obstacle is its sparsification budget. Section~6 and Appendix~A
trade approximately $2^{\eta N}$ branches for an occurrence bound
\begin{equation}
 B=O_k\!\left([\eta^{-1}\log(2/\eta)]^{k-1}\right).
 \label{eq:sparsification-budget}
\end{equation}
For fixed $k$, the chosen $\eta$ leaves $B$ polylogarithmic, as required by the
subsequent sparse-system analysis. The explicit construction uses a parameter
$M_0\ge\Gamma=2(3^{k-1}-1)$ and the budget
$B=2\sum_{j=0}^{k-1}M_0^j$. Substituting $K=\Theta(\log N)$ into this same choice
gives $\log B=\Omega(K^2)$, outside any fixed polylogarithmic budget.
This concerns the construction's guaranteed budget, not a lower bound on every
leaf's actual occurrences. It nevertheless identifies why the existing proof
does not transfer directly: the cost of organizing the constraints is no longer
controlled in the regime needed by the remaining steps. A new sparsification or
representation argument would be required to obtain a fixed proportional
exponent saving in the original number of variables.

The preceding analysis motivates a hypothesis about the fixed promise family
produced by the self-referential construction of Theorem~\ref{thm:flip}. Let
$\mathcal P=\bigcup_N\mathcal P_N$, where $\mathcal P_N$ consists of its
$N$-variable uniquely satisfiable mother formulas and their unsatisfiable
self-referential companions, within the stipulated logarithmic-width parameter
regime. The family is fixed independently of the solving algorithm. Inputs have
encoding length $N^{O(1)}$ and consist of the individual formula, without its
solution or a label identifying its role in the pair. The task is to decide
satisfiability on this promise.

We propose a \emph{dual-evidence criterion for hardness hypotheses}:
a well-motivated hypothesis should be supported both by intrinsic
structural mechanisms and by the observed limits of available
algorithmic methods. The structural evidence should identify what
obstructs the acquisition, organization, and reuse of information,
and explain why these obstructions are expected to prevent the
claimed acceleration. The algorithmic evidence should establish
that no known method achieves this acceleration on the specified
instance family, in the specified computational model and resource
regime, and examine why existing strategies fall short. These two
requirements connect an internal explanation of hardness with its
external algorithmic manifestations. They provide a standard for
justifying a hypothesis, rather than a substitute for proving it.

The preceding analysis motivates \SRKESH{} according to this
criterion. Solution independence and local--global information
gaps provide structural grounds for expecting resistance to
acceleration, while the algorithmic strategies examined above
provide no established fixed proportional saving in the exponent
for the stated promise family. The hypothesis makes this expectation
precise and subjects it to further structural analysis and
algorithmic testing.

\begin{assumption}[Self-referential $K$-SAT exhaustive-search hypothesis]
\label{ass:srk-esh}
For every $\delta\in(0,1)$, no uniform deterministic algorithm $A$ that
decides satisfiability correctly on every input in $\mathcal P$ has
\[
 \max_{F\in\mathcal P_N}T_A(F)
 =O^*\!\left(2^{(1-\delta)N}\right)
\]
for all sufficiently large $N$. The notation $O^*$ suppresses factors
polynomial in the encoded input length. We call this assumption \SRKESH.
\end{assumption}

\SRKESH{} rules out a fixed proportional saving in the exponent, while allowing
$2^{N-o(N)}$ running times and substantial local improvements. It is a
deterministic worst-case hypothesis on a fixed promise family, not an assertion
that most sampled instances are hard or that algorithms must literally enumerate
assignments. Its growing-width scope differs from fixed-width ETH and SETH;
a subexponential algorithm for every fixed width does not by itself resolve it.
The motivating question is whether the cost of constructing and reusing global
relations can be made small enough to overcome the limited transfer and local
amplification of information described above.

\subsection{\texorpdfstring{\SRKESH}{SRK-ESH} as a Projection of G\"{o}del Incompleteness}
\label{sec:srk-projection}

Section~\ref{sec:independence-acceleration} motivated \SRKESH{}
through a dual-evidence criterion combining structural mechanisms
with the observed limits of available algorithmic strategies. We now interpret this
hypothesis in terms of the self-referential local--global gap. G\"{o}del's
construction exposes limitations of sufficiently expressive formal theories
through self-reference~\cite{godel1931}. In the finite setting studied here,
solution-dependent blocking produces formulas whose global satisfiability can
differ while the prescribed local evidence agrees. \SRKESH{} expresses a
quantitative conjecture about the computational cost of overcoming this gap
on the fixed promise family $\mathcal P$.

Two complementary lines of research illuminate this interpretation. The first,
developed in Section~6.3, studies the resources required for formal refutation.
The exponential lower bounds there for Resolution, Polynomial Calculus and PCR,
Sum-of-Squares, bounded-product-width Positivstellensatz, and semantic Cutting
Planes hold under their stated conditions. They rely on expansion of the
unsigned incidence structure, with an additional partition condition for
semantic Cutting Planes, and are preserved by the support-preserving
self-referential flip. These results give concrete manifestations of exponential
resource requirements in the specified deductive models. They support the
structural motivation for \SRKESH{} without identifying a $2^{\Omega(N)}$
proof-size bound with a $2^{N-o(N)}$ unrestricted running-time bound.

The second line, pursued by Xu and Zhou~\cite{xu2025,xu2025b}, examines
self-referential Model RB instances and their encodings through a
divide-and-conquer framework. Its analysis relates the construction of
self-reference to the ability of the stipulated operations to accelerate
search. This provides an operational perspective complementary to the
proof-complexity results. The instance families and permitted operations differ,
so these analyses motivate the question of exhaustive-search complexity on
$\mathcal P$ rather than supply a proof of \SRKESH{} for that family.

The projection has three levels. Solution independence enables the construction
to use its own unique solution in a support-preserving blocking step. Structural
irreducibility then expresses the incompleteness of the specified local views
for global SAT/UNSAT determination. Finally, \SRKESH{} proposes a quantitative
boundary for uniform algorithms given the complete input: no fixed proportional
saving in the exponent on the entire promise family. The first two levels
provide the finite self-referential setting and its established local
limitations; the third asks how far those limitations persist when computation
actively constructs and reuses global relations.

This viewpoint makes the candidate space itself a reference scale for
complexity. The distinction between polynomial and superpolynomial computation
addresses one scale, while the distinction between $2^N$ and
$2^{(1-\delta)N}$ asks how much of an exponential search can be avoided. The
latter question complements the finer resource analyses developed in
parameterized and fine-grained complexity~\cite{downey1999,chen2006,cygan2016,will2015}.
For the present construction, \SRKESH{} gives this question an explicit instance
family, size parameter, and correctness requirement.

The two meta-obstacles discussed in Section~6.1 also clarify the remaining
task. The \emph{Self-Recognition Obstacle} is made concrete by constructing
opposite global states compatible with the permitted local evidence; explicit
self-reference makes the limitation of that evidence accessible to analysis.
The \emph{Static-Dynamic Obstacle} concerns the additional work needed to turn
such a static limitation into a running-time statement. That work requires a
specified computational model and control of the costs of acquiring, combining,
and reusing information.

In this sense, \SRKESH{} can be understood as a projection of G\"{o}del-style
incompleteness onto finite, resource-bounded computation. It gives a quantitative
expression to the gap between restricted local evidence and a global semantic
task, and asks whether computational organization can close that gap with a
fixed proportional exponent saving. Establishing or refuting the hypothesis
would determine how far the self-referential obstruction survives beyond the
local observation and proof models analyzed here.

\subsection{The Limits of Machine Learning}

Computational complexity and machine learning both investigate finite computation but diverge in resource focus: complexity prioritizes dynamic execution time, while machine learning optimizes informational throughput. 
By modeling computation on a foundation of pure information, this work establishes a unified framework bridging complexity theory and statistical learning. 
Under this lens, our structural irreducibility theorem reveals that evaluating localized sub-structures is  insufficient to bound a system's global configuration.

This information-theoretic boundary mirrors the core vulnerability of deep learning architectures like Large Language Models (LLMs). 
These models operate as lossy semantic compressors, sampling probability distributions from localized statistical patterns within high-dimensional spaces. 
However, conditional Kolmogorov complexity reveals a theoretical limit to this statistical approach when restricted to bounded local observations. Under the sampling and observation conditions of Section 6.2, the unique solution typically retains $N-o(N)$ bits of conditional description complexity given the local transcript and the fixed querying and decoding procedures. Correspondingly, any fixed decoder restricted to these observations recovers the unique solution with probability at most $2^{-N+o(N)}$.

Conceptually, this shifts the focus from temporal execution traces to the descriptive information required to recover a global solution from local observations. For artificial intelligence, it implies that models fixed independently of the target instance and constrained to the same observation regime cannot compensate for missing instance-specific information merely by increasing their descriptive expressivity.
Reflecting this, our recent work~\cite{peng2026} shows that the scaling limits of contemporary AI architectures are ultimately governed by the exact same algorithmic compression bottlenecks discovered here.

\section{Conclusion}

For over half a century, computational complexity theory has predominantly viewed hardness through a temporal lens, measuring the difficulty of problems in terms of dynamic runtime steps. Yet, this traditional focus on execution time obscures a deeper foundational question: is computational intractability merely a consequence of mechanical execution speeds, or is it the physical manifestation of an inescapable logical necessity? 

Historically, when Kurt G\"{o}del~\cite{godel1931} and Alan Turing~\cite{turing1936} mapped the ultimate boundaries of formal reason and computability in response to Hilbert's foundational questions of completeness, consistency, and decidability~\cite{hilbert1930}, they bypassed operational metrics entirely.
Instead, they deployed explicit self-referential constructions to prove the absolute limits of mathematics and computation, exposing the inherent structural blindness of any formal system attempting to evaluate its own logical limits from inside. In infinite domains, this structural blindness manifests as absolute incompleteness and undecidability. When cast into finite combinatorial domains, this barrier does not vanish; rather, it condenses into a tangible, self-referential bottleneck governing resource-bounded computation. 

In this work, we directly extend the foundational lineage of G\"{o}del and Turing into discrete computing architectures, establishing a formal, finite analogue of the incompleteness theorems. By mapping G\"{o}del incompleteness onto finite Boolean logic through a self-referential $K$-SAT ensemble with logarithmic clause width $K = O(\log N)$, we demonstrate that computational complexity is governed by the exact same mechanism: computational hardness is a binary logical necessity dictated by self-referential structures, manifesting as absolute irreducibility where a localized part cannot determine a global property.

This perspective adds an information-centered question to the study of
complexity: how much can an algorithm save, relative to the candidate space,
by organizing and reusing the constraints of an instance? The local--global
gap and the proof-complexity lower bounds give precise answers in the models
studied here. The comparison with short-clause algorithms identifies the
additional role of width in turning partial information into forced values
or useful correction steps. Together, these results motivate the study of
global information organization beyond the local mechanisms themselves.

\SRKESH{}, formulated in Section~\ref{sec:independence-acceleration}, makes
this direction concrete. It excludes a uniform deterministic satisfiability
algorithm with a fixed proportional exponent saving on the self-referential
promise family, while allowing $2^{N-o(N)}$ time and substantial local
improvements. The hypothesis concerns growing-width instances and is distinct
from the historical fixed-width formulation of SETH. The dual-evidence
criterion proposed here makes explicit how such a hypothesis should be
motivated: an internal explanation of hardness should be accompanied
by external algorithmic evidence, with both addressing the same
computational claim. As discussed in
Section~\ref{sec:srk-projection}, it provides a quantitative computational
projection of the G\"{o}del-style local--global gap. The next step is to relate
the established static independence properties to the cost of constructing
and reusing global relations in concrete computational models. This connects
the present construction with the broader study of self-reference and
information irreducibility~\cite{xu2025,li2026,zhou-ds2026,zhou2026}.

Future research will focus on two directions: first, extending self-reference 
and solution-independence to a broader spectrum of combinatorial problems, operational frameworks, 
and machine learning models, thereby delineating the fundamental boundaries and phase 
transitions of these systems; and second, integrating our self-referential invariants 
into the framework of structural information~\cite{li2016}, which extends classical 
information theory to provide a precise mathematical foundation for complex data analysis
and artificial intelligence~\cite{li2024}.

Ultimately, whether manifested as G\"{o}del incompleteness, Turing undecidability, or the structural irreducibility established in this work, the fundamental reality remains the same: local syntactic deduction is inherently blind to global semantic truths. Following this line, we arrive at an invariant law governing the nature of complexity: \emph{The ultimate barrier is not the inability to find an answer, but the mathematical impossibility of ever being certain through local evaluation that what we see is the global truth.}

\end{document}